\documentclass[11pt]{article}
\usepackage{amssymb}
\usepackage{microtype}
\usepackage{caption}
\usepackage{graphicx,subfig}
\usepackage{amsmath}
\usepackage{float}
\usepackage[letterpaper,margin=1in]{geometry}   
\usepackage{xcolor} 
\usepackage{amsthm}
\usepackage{algorithmic}
\usepackage{algorithm}
\usepackage{xurl}
\usepackage[colorlinks,
linkcolor=blue,
anchorcolor=blue,
citecolor=blue]{hyperref}
\usepackage[numbers]{natbib}
\usepackage{enumitem}
\usepackage{makecell}

\newtheorem{theorem}{Theorem}
\newtheorem{lemma}{Lemma}

\newtheorem{assumption}{Assumption}

\newcommand{\red}[1]{#1}
\newcommand{\blue}[1]{#1}
\newcommand\fro[1]{\| #1 \|_{\rm{F}}}
\newcommand\op[1]{\| #1 \|}
\newcommand\nuc[1]{\| #1 \|_{*}}

\newcommand\ltwo[1]{\| #1 \|_{\ell_2}}
\newcommand\linf[1]{\| #1 \|_{\ell_{\infty}}}

\newcommand{\inp}[2]{\langle #1,#2\rangle}
\newcommand{\psitwo}[1]{\big\| #1 \big\|_{\psi_2}}
\newcommand\twoinf[1]{\| #1 \|_{2,\infty}}

\def\calI{{\mathcal I}}

\def\calL{{\mathcal L}}
\def\calM{{\mathcal M}}

\def\calP{{\mathcal P}}

\def\calR{{\mathcal R}}

\def\calY{{\mathcal Y}}
\def\calZ{{\mathcal Z}}

\def\bcalC{{\boldsymbol{\mathcal C}}}

\def\bcalE{{\boldsymbol{\mathcal E}}}

\def\bcalT{{\boldsymbol{\mathcal T}}}

\def\bcalX{{\boldsymbol{\mathcal X}}}

\def\EE{{\mathbb E}}

\def\KK{{\mathbb K}}

\def\MM{{\mathbb M}}

\def\OO{{\mathbb O}}
\def\PP{{\mathbb P}}

\def\RR{{\mathbb R}}

\def\TT{{\mathbb T}}

\def\a{{\boldsymbol a}}

\def\e{{\boldsymbol e}}

\def\r{{\boldsymbol r}}

\def\u{{\boldsymbol u}}
\def\v{{\boldsymbol v}}

\def\x{{\boldsymbol x}}
\def\y{{\boldsymbol y}}

\def\A{{\boldsymbol A}}
\def\B{{\boldsymbol B}}
\def\C{{\boldsymbol C}}
\def\D{{\boldsymbol D}}
\def\E{{\boldsymbol E}}

\def\G{{\boldsymbol G}}

\def\I{{\boldsymbol I}}

\def\L{{\boldsymbol L}}
\def\M{{\boldsymbol M}}
\def\N{{\boldsymbol N}}
\def\O{{\boldsymbol O}}
\def\P{{\boldsymbol P}}

\def\R{{\boldsymbol R}}
\def\S{{\boldsymbol S}}

\def\U{{\boldsymbol U}}
\def\V{{\boldsymbol V}}
\def\W{{\boldsymbol W}}
\def\X{{\boldsymbol X}}
\def\Y{{\boldsymbol Y}}
\def\Z{{\boldsymbol Z}}

\def\sfC{\mathsf C}

\def\sfR{\mathsf R}

\def\bdelta{\boldsymbol{\delta}}
\def\beps{\boldsymbol{\epsilon}}
\def\bSigma{\boldsymbol{\Sigma}}

\def\rank{\textsf{rank}}
\def\svd{\textsf{SVD}}

\def\var{\textsf{Var}}
\def\hat{\widehat}
\def\tilde{\widetilde}
\def\trim{\textsf{Trim}}
\def\trunc{\textsf{Trunc}}

\def\incoh{\textsf{Incoh}}

\def\approach{\textsf{TCSI }}
\hypersetup{
  pdftitle={Tensor Completion using Subspace Information},
  pdfauthor={Jingyang Li and Michael K. Ng},
  pdfkeywords={Low-rank tensor completion, Side information, Tucker decomposition, Non-convex optimization, High-dimensional statistics},
  bookmarksnumbered=true
}
\title{\Large\bfseries Tensor Completion using Subspace Information}
\author{Jingyang Li$^{1}$ \qquad Michael K. Ng$^{2}$\thanks{M. Ng's research is supported by the
GDSTC: Guangdong and Hong Kong Universities ``1+1+1'' Joint Research Collaboration Scheme UICR0800008-24,
National Key Research and Development Program of China under Grant 2024YFE0202900,
RGC GRF 12300125.}\\[0.6em]
{\small $^{1}$Department of Statistics and Data Science, Fudan University}\\
{\small $^{2}$Department of Mathematics, Hong Kong Baptist University}}
\date{}

\begin{document}
\maketitle

\begin{abstract}
		Tensor completion has attracted significant attention in both applications and theoretical research. 
		Under standard uniform sampling, existing polynomial-time guarantees generally require more observations than the number of degree of freedom, motivating the study of a possible statistical-to-computational gap in highly missing regimes.
		Fortunately, in many practical scenarios, side information is available, which can provide valuable insights to mitigate these challenges.
		In this paper, we introduce an algorithm called \underline Tensor \underline Completion using \underline Subspace \underline Information (\textsf{TCSI}) that incorporates side information through an estimated subspace. Our approach first extracts the subspace from the available side information and then reformulates tensor completion as a matrix regression problem. \red{We provide a theoretical analysis showing that, when accurate subspace information is available, the required sample complexity is reduced to nearly linear order in the uncoupled ambient dimensions, removing the coupled-mode dimension from the leading term. 
		Leveraging the estimated subspace information, we obtain a less stringent sufficient signal-to-noise ratio requirement than those in several existing passive-uniform-sampling guarantees. Under additional mild conditions, we obtain a sharper statistical error bound.} 
		Our theoretical findings are supported by numerical simulations.
		\red{We apply \approach to the reconstruction of global Total Electron Content (TEC) maps and observe lower reconstruction errors than the compared methods in our experiments.}
	\end{abstract}
	
	\noindent%

	\section{Introduction}
	Low-rank tensor completion aims to fill in missing entries of a tensor, where only a small portion of its entries is known.
	It has found wide applications in various domains, including imaging and computer vision \citep{li2010tensor,liu2012tensor,bengua2017efficient}, seismic data analysis \citep{kreimer2013tensor}, space weather data analysis \citep{sun2023complete}, recommendation system \citep{karatzoglou2010multiverse,ibriga2022covariate}, and more.
	At first glance, tensor completion is a natural generalization from the well-explored low-rank matrix completion \citep{candes2010power,recht2011simpler}. However, tensor completion presents additional statistical and computational challenges.
	For instance, the nuclear norm serves as a convex relaxation of the rank function, and is widely used in matrix completion. In contrast, the convex relaxation of tensor rank leads to the tensor nuclear norm, which poses computational challenges as it is NP-hard to compute \citep{hillar2013most,yuan2016tensor}.
	Another challenge in tensor completion is the discrepancy between the degree of freedom of the low-rank tensor and the sample size required by polynomial-time algorithms. Ideally, one would expect the sample complexity at the level of degree of freedom to be sufficient. 
	\red{Existing polynomial-time guarantees for methods such as gradient descent \citep{cai2019nonconvex}, power iteration \citep{xia2021statistically}, scaled gradient descent \citep{tong2022scaling}, and Riemannian gradient descent \citep{cai2022provable} generally require sample sizes above the degree-of-freedom scale under uniform random sampling.} 
	\red{Estimators based on tensor norms provide important statistical benchmarks \cite{barak2016noisy,ghadermarzy2019near}. Specifically, Ghadermarzy et al. \cite{ghadermarzy2019near} showed that M-norm and max-qnorm-constrained estimators attain optimal dependence on the ambient dimension and nearly minimax-rate-optimal error bounds under their model, while noting that the rank dependence need not be optimal. They also noted that no polynomial-time method is known for the M-norm-constrained estimator.}
	
	\red{Taken together, these results motivate the study of the trade-off between statistical and computational requirements under passive uniform sampling.} 
	\red{A complementary line of work changes the observation model through structured or adaptive sampling. Zhang \citep{zhang2019cross} proposed CROSS, which uses a designed cross-shaped pattern consisting of a dense core and sparse arms and attains the degree-of-freedom sample complexity for Tucker-rank tensors. Haselby et al. \citep{haselby2023tensor} proposed Tensor Sandwich, which adaptively selects entries and attains near-linear sample complexity for low-CP-rank tensors. These methods are important benchmarks: they show that alternative sampling designs can avoid the ambient-dimensional bottleneck encountered under passive uniform sampling. CROSS is particularly suitable when the observation pattern can be designed, whereas Tensor Sandwich applies when observations can be selected adaptively.} 
	
	\red{Our work addresses a complementary setting in which the observed entries are sampled passively and uniformly, while external subspace information is available. Under uniform sampling, \cite{barak2016noisy} provides conditional evidence of a computational barrier at information-theoretic sample sizes. Related barriers have also been studied in tensor regression \citep{shen2022computationally}, tensor PCA \citep{zhang2018tensor}, and tensor clustering \citep{luo2022tensor}.}
	
	In practice, it is quite common that the missing percentage of the tensor is very high and the sample size condition is thus violated. 
	Fortunately, with the ability to access massive amounts of data as a result of recent technological advances, we can get access to data from multiple sources. 
	For example, in a restaurant recommendation system, in addition to the ratings users make for certain type of restaurant of a specific meal (i.e., breakfast, lunch, dinner), friendship networks of these users are also available. 
	In this case, we have a rating tensor (user $\times$ category $\times$ meal) together with a friendship network (user $\times$ user). 
	It is broadly assumed these datasets share the same subspace along the user dimension. 
	And this property is encoded in the loss function by using the same subspace representation for the coupled mode \citep{acar2011all,acar2013understanding}. 
	Later \cite{zhou2017tensor} proposed a Riemannian conjugate gradient descent algorithm to solve the tensor completion problem in the presence of side information.
Wimalawarne et al. \cite{wimalawarne2018convex} proposed convex
low-rank-induced norms for coupled matrix--tensor data based on mode
unfoldings.
	\red{Much of the earlier literature emphasizes algorithm design, while theoretical guarantees under explicit sampling models are less common. Yu and Xi \cite{yu2022tensor} showed that, for low orthogonal CP-rank tensors, one weight vector for each mode that is not orthogonal to any latent factor along that mode is sufficient for a consistent estimator using $O(d^{1+\varepsilon})$ uniformly sampled observations for any fixed $\varepsilon>0$, with extensions beyond the orthogonal setting.} 
	\red{Ibriga and Sun \cite{ibriga2022covariate} proposed the COSTCO alternating-minimization method, which jointly uses tensor and covariate information and establishes improved estimation for the component associated with the shared mode under mild conditions.} 
	These results demonstrate the potential of leveraging side information to enhance tensor completion performance in real-world applications.

	\subsection{Our Contributions}
	\red{We study noisy Tucker-rank tensor completion under passive uniform sampling when an approximate subspace estimate is available along one coupled mode.} Without loss of generality, we assume the low-rank tensor $\bcalT^*\in\RR^{d_1\times d_2\times d_3}$ has Tucker rank $\r = (r_1,r_2,r_3)$ and the side information is coupled along the first mode. The algorithmic construction can also be extended to settings in which side information is coupled along more than one mode. 
	Our approach consists of two parts: first, we extract the subspace $\U_s\in\RR^{d_1\times r_1}$ along the coupled dimension from the side information in the format of either a matrix or a tensor; second, we reformulate the tensor completion problem into a matrix regression problem and use the Riemannian gradient descent algorithm to solve the low-rank matrix regression. We refer to the resulting algorithm as \underline Tensor  \underline Completion using \underline Subspace \underline Information (\textsf{TCSI}).
	
	\red{With moderately accurate subspace information, our bound reduces the dependence of the required sample size on the ambient dimensions. In particular, in the fixed-rank, bounded-incoherence, and bounded-condition-number regime, it requires $\widetilde O(d_2\vee d_3)$ samples\footnote{Here $\widetilde O(\cdot)$ hides logarithmic factors.}. We also use side information in the spectral initialization, so the main theorem covers both the side-information-aided initialization and local convergence without presupposing an externally supplied warm start for the reduced matrix-regression problem. Under the stated assumptions, the initialization and local convergence guarantees have reduced ambient-dimensional sample and signal-to-noise requirements, and when the side information is sufficiently accurate, the statistical-error bound is also smaller in regimes with $d_1\gg d_2+d_3$.}
	
	\red{These results complement existing theoretical guarantees for tensor completion with side information. Yu and Xi \cite{yu2022tensor} obtain nearly linear sample complexity using one weak weight vector for every mode, primarily under a low orthogonal CP-rank model; \approach instead studies noisy Tucker-rank tensors with an approximate subspace along only one mode. COSTCO \cite{ibriga2022covariate} studies a coupled CP model with a covariate matrix and establishes improved estimation for the component associated with the shared mode. Under bounded rank, incoherence, condition numbers, and tensor-to-covariate signal ratios, its sample-size requirement is $\widetilde O(d^{3/2})$ in the balanced setting. Our analysis provides a unified guarantee for side-information-aided initialization, local convergence, and overall tensor error in the one-mode subspace setting.}
	
	\red{Several representative non-unfolding guarantees without side information summarized in Tables \ref{table:literature} and \ref{table:literature:nonTucker} contain the ambient-dimensional term $\widetilde O(\max\{d_1,d_2,d_3\}+\sqrt{d_1d_2d_3})$ under their respective assumptions. This comparison quantifies the benefit of the additional subspace information rather than an assumption-free dominance. The improvement concerns the dependence on the ambient dimensions and comes with a less favorable dependence on rank. Assuming $d_1=d_2=d_3=d$ and $r_1=r_2=r_3=r$, our theorem gives $\widetilde O_{\mu,\kappa}(dr^8)$,\footnote{The subscript $\mu,\kappa$ indicates that the hidden constant may depend on the incoherence parameter $\mu$ and the condition number $\kappa$.} whereas \cite{xia2021statistically} gives $\widetilde O_{\mu,\kappa}(d^{3/2}r^4)$. Thus, the gain is most relevant when $d$ is sufficiently large relative to $r$ and informative side information is available.} 

	\begin{table}[htbp]
		\centering
		{\renewcommand{\arraystretch}{0.85}\renewcommand\theadfont{\scriptsize}
			\begin{tabular}{c||c|c|c|c}
				\hline
				\tiny Method/estimator & \tiny Sufficient sample size & \tiny Noise? & \tiny Side info.? & \tiny Statistical error \\		
				\hline
				\hline
				\thead{Unfolding \\ \citep{huang2015provable}} &\tiny$(d_1d_3\vee d_2) \log^2d$ &No&No& NA\\ 
				\hline
				\thead{Tensor nuclear norm \\ \citep{yuan2016tensor}} &\tiny$\sqrt{d_1d_2d_3}\log^3 d + d\log^{3}d$ &No &No & NA\\
				\hline
				\thead{Grassmannian GD\\ \citep{xia2019polynomial}} &\tiny$\sqrt{d_1d_2d_3}\log^{7/2} d + d\log^{6}d$ &No&No&NA\\
				\hline
				\thead{Scaled GD \\ \citep{tong2022scaling}} & \tiny$\sqrt{d_1d_2d_3}\log^3 d + d\log^{5}d$ &No &No & NA\\
				\hline
				\thead{Riemannian \\optimization\\ \citep{wang2023implicit}} &\tiny$\sqrt{d_1d_2d_3}\log^3 d + d\log^{5}d$&No &No &NA\\
				\hline
				\thead{Power Iteration \\ \citep{xia2021statistically}} & \tiny$\sqrt{d_1d_2d_3}\log^5 d + d\log^{10}d$ & Yes & No & \tiny$\frac{d}{n} \log d\cdot\sigma^2+\frac{d}{n} \log d\cdot\frac{\lambda_{\max}^2}{d^*}$\\
				\hline
				\thead{\approach \\ (This paper)} & \tiny$\boldsymbol{(d_2\vee d_3)\log d}$& Yes &Yes &  \tiny $\frac{(d_2\vee d_3)}{n}\log d\cdot\sigma^2 + \frac{1}{d^*}\lambda_{\max}^2\delta^2$\\
				\hline 
		\end{tabular}}
		\caption{\red{Selected sufficient guarantees for low Tucker-rank tensor completion under passive uniform sampling. \approach assumes accurate subspace information, whereas the other listed methods do not, so the table is not an assumption-free ranking. For \approach, the displayed sample size counts only the observed entries of the target tensor, conditional on an estimated subspace $\U_s$ satisfying the stated accuracy conditions.} We assume $\r,\kappa,\mu,\mu_s\asymp O(1)$, let $d=\max\{d_1,d_2,d_3\}$, and define $\delta=\fro{\U_s\U_s^\top-\U^*\U^{*\top}}$. 
			The \approach row reports the simplified fixed-rank regime in which the third term of Theorem \ref{thm:main} is dominated. 
			\red{When $d_1=d_2=d_3=d$, the listed non-unfolding baselines scale as $\widetilde O(d^{3/2})$, whereas \approach scales as $\widetilde O(d)$ under its additional assumption.}}
		\label{table:literature}
	\end{table}

	\begin{table}[htbp]
		\centering
		\begin{tabular}{c||c|c|c}
			\hline
			Method/estimator & Tensor format & Side info.? & Sample size \\		
			\hline 
			\thead{ Alternating minimization\\ \citep{jain2014provable}} & CP & No & $d^{3/2}\log^4 d$\\
			\hline
			\thead{Sum-of-squares \\ \citep{barak2016noisy}} & CP & No & $d^{3/2}\log^4 d$ \\
			\hline
			\thead{Vanilla gradient descent \\ \citep{cai2019nonconvex}} & CP &No & $d^{3/2}\log^4 d$\\
			\hline
			\thead{Riemannian optimization \\ \citep{cai2022provable}}& Tensor-train &No &$d^{3/2}\log^5 d$ \\
			\hline 
			\thead{Weak-side-information estimator \\ \cite{yu2022tensor}} & \thead{Orthogonal CP \\ (main result)} & Yes & $d^{1+\varepsilon}$\\
			\hline
			\thead{COSTCO \\ \citep{ibriga2022covariate}} &CP&Yes& $d^{3/2}\log^2 d$\\ 
			\hline
		\end{tabular}
		\caption{\red{Selected passive uniform-sampling guarantees for low-rank tensor completion under different tensor formats. Here we assume $\r, \kappa, \mu\asymp O(1)$ and $d_1=d_2=d_3=d$ for ease of presentation. The assumptions, including access to side information, differ across methods. For Yu and Xi, $\varepsilon>0$ is arbitrary and the displayed rate corresponds to the main low orthogonal CP-rank result with one weak weight vector per mode. For COSTCO, the displayed order additionally assumes bounded tensor-to-covariate signal ratios.}}
		\label{table:literature:nonTucker}
	\end{table}
	
	The subspace estimated from the side information is typically biased due to the existence of noise. Although our theoretical analysis does not require the exact subspace, the bias will result in a non-vanishing error in our final estimate of the original tensor. 
	\red{When the subspace is sufficiently accurate, its contribution is dominated by the observation-noise term, and the resulting rate reaches the relevant degree-of-freedom scale of the reduced
		matrix-regression problem up to logarithmic factors under the stated assumptions.} 
	\red{We also give common examples illustrating regimes in which the required side-information accuracy can hold.} 
	\red{When the subspace is only moderately accurate, the final error is instead governed by the quality of the subspace estimate. Thus, \approach is an option for tensor completion when observations are limited and informative side information is available.}

	\subsection{ Organization of the Paper}
	The rest of the paper is organized as follows. Section \ref{sec:notation} introduces the notation used throughout this paper and reviews the basics of Tucker rank. 
	The problem formulation and the details about the side information, and the algorithm are described in Section \ref{sec:formulation}. 
	In Section \ref{sec:theory}, we present the main theorem regarding the convergence of our proposed algorithm. 
	We display the performances on simulations and the real data in Section \ref{sec:simulation}, Section \ref{sec:realdata} respectively. 
	All proofs and technical lemmas are relegated to the appendix. 

	\section{Notations and Preliminaries}\label{sec:notation}
	In this section, we introduce the notations used throughout the paper and review the background of tensors. 
	Throughout this paper, we shall use the bold calligraphic letters (e.g. $\bcalT,\bcalC$) to denote tensors, the bold capital letters (e.g. $\M$) to denote matrices, blackboard bold-face letters (e.g. $\RR,\TT$) for sets, and the lower case bold-face letters ($\x,\y$) to denote vectors.
	The $j$-th canonical basis vector is denoted by $\e_j$, and we omit the ambient space it lies in whenever the context is clear.
	For a positive integer $d$, denote $[d] := \{1,\cdots,d\}$. 
	Let $\fro{\cdot}$ denote the Frobenius norm of tensors or matrices. We use $\|\cdot\|_{\ell_p}$ to denote the $p$ norm of vectors for $0<p\leq \infty$ and $\|\cdot\|_{\ell_0}$ to represent the number of nonzero entries. The notations $C,C_1,\cdots$ are reserved for some positive and absolute constants which do not depend on the related parameters of the problem, but their actual values may change from line to line. We denote $d= \max\{d_1,d_2,d_3\}$, and $d^*=d_1d_2d_3$, $r^*=r_1r_2r_3$. 
	Let $\OO_{d,r} = \{\U\in\RR^{d\times r}: \U^\top\U = \I_r\}$ be the collection of column orthonormal matrices. 
	For a matrix $\M$, we use $\sigma_i(\M)$ to represent the $i$-th largest singular value of $\M$. 
	
	Let $\bcalT\in\RR^{d_1\times d_2\times d_3}$ be a third order tensor. 
	The standard basis in $\RR^{d_1\times d_2\times d_3}$ is denoted by $\{\bcalE_{\omega}:\omega\in[d_1]\times [d_2]\times [d_3]\}$. 
	For $\omega=(i_1,i_2, i_3)\in[d_1]\times [d_2]\times [d_3]$, we use $[\bcalT]_{\omega}$ or $[\bcalT]_{i_1,i_2,i_3}$ as the entry of $\bcalT$. 
	The $j$-th matricization (also called unfolding) $\calM_j:\RR^{d_1\times d_2\times d_3}\rightarrow \RR^{d_j\times d_j^-}$ with $d_j^- = d^*/d_j$ is a linear reshape mapping so that, for example $[\calM_1(\bcalT)]_{i_1,(i_2-1)d_3+i_3} = [\bcalT]_{i_1,i_2,i_3}$ for all $i_j\in[d_j]$. The collection 
	$$\rank(\bcalT) = \big(\rank(\calM_1(\bcalT)),\rank(\calM_2(\bcalT)),\rank(\calM_3(\bcalT))\big)$$ 
	is called the multi-linear rank or the Tucker rank of $\bcalT$. 
	Given a matrix $\W\in\RR^{p\times d_j}$ for any $j\in[3]$, the multi-linear product, denoted by $\times_j$, between $\bcalT$ and $\W$ is defined by (we take $j=1$ for an example and $j=2,3$ can be similarly defined):
	$$[\bcalT\times_1\W]_{l,i_2,i_3} = \sum_{i_1=1}^{d_1}[\bcalT]_{i_1,i_2,i_3}[\W]_{l,i_1},\quad l\in[p], i_j\in[d_j]. $$
	If $\rank(\bcalT) = (r_1,r_2,r_3)$, then there exists $\U\in\OO_{d_1,r_1}$, $\V\in\OO_{d_2,r_2}$, $\W\in\OO_{d_3,r_3}$, and $\bcalC\in\RR^{r_1\times r_2\times r_3}$, such that 
	$$\bcalT = \bcalC\times_1\U\times_2\V\times_3\W. $$ 
	This is referred to as the Tucker decomposition of the tensor, and it can be obtained using high order singular value decomposition (HOSVD), whose details are described in Algorithm \ref{alg:hosvd} in Appendix \ref{app:alg}.

	\section{Problem Formulation}\label{sec:formulation}
	The goal of tensor completion is to recover an underlying tensor from a small subset of the observed entries. We shall denote $\bcalT^*\in\RR^{d_1\times d_2\times d_3}$ the underlying tensor of Tucker rank $\r = (r_1,r_2,r_3)$. 
	Let its Tucker decomposition be $\bcalT^* = \bcalC^*\times_1 \U^*\times_2\V^*\times_3 \W^*$. 
	We consider a common sampling scheme \citep{jain2014provable,ibriga2022covariate,xia2021statistically,cai2022provable}, which assumes each observation is uniformly and randomly sampled from the original tensor:
	\begin{align}\label{model}
		y_i = [\bcalT^*]_{\omega_i}  +\epsilon_i,\quad  \omega_i\sim \text{Unif}([d_1]\times [d_2]\times [d_3]), \quad i = 1,\cdots, n.
	\end{align}
	And $\epsilon_i$ is the random noise. We remark that this sampling scheme is equivalent to assuming that the tensor entries are missing completely at random.
	In order to recover $\bcalT^*$ from the observations, it is natural to find a low rank tensor that is consistent with the observed entries:
	\begin{align}\label{prob:TC}
		\min_{\bcalT}\frac{1}{2}\sum_{i=1}^n (y_i - [\bcalT]_{\omega_i})^2, \quad \text{s.t.}~\rank(\bcalT)\leq \r. 
	\end{align}
	Here, the inequality $\leq$ between two tuples is defined element-wise. 
	Due to the low Tucker rank constraint, this problem is highly non-convex and is only solvable locally. Once a good initial estimator is provided, a particularly popular class of algorithms is based on the projected gradient descent \citep{jain2010guaranteed,chen2019non}, which consist of gradient update and retraction back to the low rank manifold using HOSVD.
	\red{Projected-gradient methods can require SVDs of full ambient matrices. Riemannian gradient descent reduces this cost by exploiting low-rank tangent-space structure and has been applied to tensor completion \citep{steinlechner2016riemannian,cai2022provable}; it nevertheless remains a nonconvex method requiring a suitable initialization. We discuss this algorithm in more detail below. }
	
	\subsection{Incoherence and Condition Number}
	Tensor completion is extremely hard, for example, if $\bcalT^*$ has only one non-zero entry. In this case, it is impossible to recover $\bcalT^*$ unless this entry is observed. Thus a common assumption made to ensure the information $\bcalT^*$ carries is spread evenly across its entries is the incoherence condition \citep{xia2019polynomial,xia2021statistically}. The incoherence for an column orthonormal matrix $\U\in\RR^{d\times r}$ is defined as 
	\begin{align*}
		\incoh(\U)=\frac{d}{r}\max_{i\in[d]}\ltwo{\U^\top\e_i}^2. 
	\end{align*}
	Consider a matrix with singular value decomposition $\M = \L\S\R^\top$, and a tensor with Tucker decomposition $\bcalT = \bcalC\times_1\U\times_2\V\times_3\W$, where the factor matrices $\L,\R$ and $\U,\V,\W$ are column-orthonormal. The incoherence is respectively defined as 
	\begin{align*}
		\incoh(\M) &= \max\big\{\incoh(\L), \incoh(\R)\big\}, \\ 
		\incoh(\bcalT) &= \max\big\{\incoh(\U), \incoh(\V),\incoh(\W)\big\}.
	\end{align*}
	It is easy to verify that the incoherence of matrix/tensor is determined by the subspaces spanned by the factors and is independent of the specific decomposition up to permutation in the SVD and orthogonal transformation of the factors in the Tucker decomposition.
	In this paper, we will focus on the underlying tensor whose incoherence is bounded by $\mu > 0$, that is $\incoh(\bcalT^*)\leq\mu.$
	The condition number of the tensor $\bcalT^*$ is defined by $\kappa = \lambda_{\max}/\lambda_{\min}$, where $\lambda_{\max} = \max_{i=1}^3\op{\calM_i(\bcalT^*)}$, and $\lambda_{\min} = \min_{i=1}^3\sigma_{r_i}\big(\calM_i(\bcalT^*)\big)$.

	\subsection{Tensor Completion with Side Information}
	In addition to the observed entries from the original tensor, side information describing the features along certain tensor mode is also attainable in many cases. 
	In this paper, we assume without loss of generality the side information and the tensor are coupled along the first mode. 
	We extract the subspace information $\U_s\in\OO_{d_1, r_1}$ from the side information. We will list some common ways to extract the subspace from various types of side information in the next section. 
	The accuracy of side information is jointly measured under $\fro{\cdot}$ and $\twoinf{\cdot}$, and we denote $\mu_s$ to be the incoherence of $\U_s$: 
	\begin{align*}
		\delta := \fro{\U_s\U_s^\top - \U^*\U^{*\top}}, \quad
		\gamma := \frac{\sqrt{d_1}}{\delta}\twoinf{\U_s - \U^*\R},\quad
		\mu_s := \incoh(\U_s),
	\end{align*}
	where $\R = \arg\min_{\R\in\OO_{r_1}}\fro{\U_s - \U^*\R}$ aligns the two matrices. 
	Here, $\R$ minimizes the Procrustes distance and is given explicitly by $\R = \O_2\O_1^\top$, where $\U_s^\top\U^* = \O_1\bSigma_{0}\O_2^\top$ is the SVD. 
	The parameter $\gamma$ captures how the difference between the subspaces is distributed across rows. It ranges from $\gamma \approx 1$ for uniformly distributed (incoherent) error to $\gamma = \sqrt{d_1}$ for row-concentrated (spiky) error. We provide a detailed derivation in Appendix \ref{app:gamma}. 
	We shall in the following denote $\bcalE_{\omega_i} = \e_{l_i}\otimes \e_{k_i}\otimes \e_{g_i}$ with $\omega_i=(l_i,k_i,g_i)$. 
	\blue{Let $\C_2^*:=\calM_2(\bcalC^*)$ denote the mode-2 matricization of the core tensor $\bcalC^*$.}
	Then \eqref{model} can be rewritten as 
	\begin{align}\label{def:y}
		y_i &= \inp{\bcalE_{\omega_i}\times_1\U_s^{\top}}{\bcalC^*\times_1\R^\top\times_2\V^*\times_3\W^*}\notag\\
        &\quad + \inp{\bcalE_{\omega_i}\times_1(\U^{*}\R- \U_s)^\top}{\bcalC^*
			\times_1\R^\top\times_2\V^*\times_3\W^*} + \epsilon_i\notag\\
		&= \inp{\underbrace{\E_{i,2}(\I_{d_3}\otimes \U_s)}_{=:\Y_i}}{\underbrace{\V^*\C^*_2(\W^{*\top}\otimes\R)}_{=:\M^*}}+  \underbrace{\inp{\bcalE_{\omega_i}\times_1(\U^{*}\R- \U_s)^\top}{\bcalC^*\times_1\R^\top\times_2\V^*\times_3\W^*}}_{=:\delta_i} + \epsilon_i,
	\end{align}
	where $\E_{i,2} = \calM_2(\bcalE_{\omega_i})$.
	Then $\rank(\M^*)= \min\{r_2,r_1r_3\} =: r$. 
	We define the linear map $\calY:\RR^{d_2\times d_3r_1}\rightarrow \RR^n$ as 
$
		\calY(\M) = \big(\inp{\Y_1}{\M},\cdots,\inp{\Y_n}{\M}\big)^{\top}. 
$
	We denote $\y = (y_1,\cdots,y_n)^\top$. 
	\blue{We denote the adjoint of $\calY$ by $\calY^*:\RR^n\rightarrow\RR^{d_2\times d_3r_1}$; for $\a=(a_1,\ldots,a_n)^\top$, it is given by $\calY^*(\a)=\sum_{i=1}^n a_i\Y_i$.}
	Then we consider the following problem:
	\begin{align}\label{prob:MR}
		\min_{\M} f(\M):=\frac{1}{2}\ltwo{\calY(\M) - \y}^2, \quad \text{s.t.~} \rank(\M)\leq r. 
	\end{align}
	Here $\rank$ refers to the matrix rank, and we shall abuse the notation when the context is clear. 
	\red{We convert the original tensor completion problem into a matrix regression problem whose ambient matrix dimension is reduced by the subspace information.} 
	\red{Although inaccuracy of $\U_s$ introduces the error terms $\{\delta_i\}_{i=1}^n$, this reformulation yields the reduced ambient-dimensional sample bound established in the next section under the stated side-information assumptions.} 

	\subsection{Algorithm}
	Now we turn our attention to solving \eqref{prob:MR}, which is a low-rank matrix regression problem. 
	Low-rank matrix recovery problem is widely studied in a lot of works but usually the sampling basis (corresponding to our $\Y_i$) is either random Gaussian matrix or standard basis in the matrix space. However, due to the nature of unfolding, our $\Y_i$ are more structured and thus tailored analysis is required. 
	For problems with low-rank constraints, a commonly used method is Riemannian gradient descent (RGrad) \citep{wei2016guarantees, wei2016guarantees2, shen2023computationally}.
	Given a warm initialization, RGrad consists of three parts: (1) computing the Riemannian gradient; (2) updating along the Riemannian gradient direction; (3) retracting back to the low rank matrix manifold. The procedure is summarized in Algorithm \ref{alg:rgrad}. 
	\begin{algorithm}
		\caption{RGrad with side information}
		\begin{algorithmic}\label{alg:rgrad}
			\STATE{\textbf{Input: }$\M_0$, side information $\U_s$, maximum iteration $l_{\max}$}
			\FOR{$l = 0, \cdots, l_{\max}-1$}
			\STATE{$\G_l = \calY^*(\calY(\M_l)-\y)$}
			\STATE{$\W_l = \trim_{\zeta}(\M_l - \eta \calP_{\TT_{l}}\G_l)$ with $\zeta = \lambda_{\max}\sqrt{\frac{\mu^2r_2r_3}{d_2d_3}}$, $\eta = \frac{d^*}{n}$}
			\STATE{$\M_{l+1} = \svd_r(\W_l)$}
			\ENDFOR
			\STATE{\textbf{Output: }$\hat\bcalT = \calM_2^{-1}(\M_{l_{\max}})\times_1\U_s$}
		\end{algorithmic}
	\end{algorithm}
	
	In the first step, the Riemannian gradient is computed by projecting the vanilla gradient $\G_l = \nabla f(\M_l)$ to the tangent space $\TT_l$ at $\M_l$ with respect to the low rank matrix manifold $\MM_r$. Let $\M_l = \L_l\S_l\R_l^\top$ be the compact rank-$r$ SVD, then the projection onto the tangent space is given by 
	\begin{align*}
		\calP_{\TT_l}(\M) = \L_l\L_l^\top\M + \M\R_l\R_l^\top - \L_l\L_l^\top\M\R_l\R_l^\top. 
	\end{align*}
	The second step involves the selection of step-size, which we will give some guidance in the next section. The final step involves the retraction back to the manifold $\MM_r$, and this can be achieved by $\svd_r$, which maps a matrix to its best rank-$r$ approximation under the Frobenius norm. 
	\red{The proposed method does not remove the rank constraint or the cost of SVD computations: each iteration includes a rank-$r$ truncated SVD. Its computational benefit is more specific since the side information reduces the matrix dimension from $d_2\times d_1d_3$ to $d_2\times r_1d_3$ when $r_1\ll d_1$. }
	
	\vspace{0.2cm}
	
	\noindent\textbf{\emph{Discussion of the Trimming Step. }}
	Algorithm \ref{alg:rgrad} involves a trimming step, which is formally defined by the operator $\trim_{\zeta}$ applied entry-wise to a matrix as:
	\begin{align*}
		[\trim_{\zeta}(\M)]_{ij} = \left\{
		\begin{aligned}
			&[\M]_{ij}, \quad |[\M]_{ij}|\leq \zeta, \\
			&\zeta\cdot\textsf{sign}([\M]_{ij}), \quad\text{otherwise}. 
		\end{aligned}
		\right.
	\end{align*}
	By explicitly capping the matrix entries at the threshold $\zeta$, this step is utilized in our theoretical analysis to control the incoherence of the iterates.
	This serves as a standard theoretical tool and has been widely adopted in the existing literature (e.g., \cite{cai2022provable, cai2022generalized}).
	While trimming is often not strictly required in practical implementations, recent theoretical advances provide pathways to bypass this step entirely, such as employing leave-one-out technique \citep{ma2018implicit,shen2023quantile,wang2023implicit}, imposing stronger initialization conditions \citep{wei2020guarantees}, or shifting to an online setting \citep{cai2023online}. 
	However, each of these alternatives introduces theoretical trade-offs: leave-one-out arguments substantially complicate the proof, an online setting fundamentally alters the problem framework.
	Because our primary objective is to quantify the sample-size effect of side information under passive uniform sampling, we retain the trimming step to keep the proof tractable.

	\subsection{Side Information Aided Spectral Initialization}
	Algorithm \ref{alg:rgrad} is only useful when starting from some point that is sufficiently close to the ground-truth. 
	Therefore the design of the initialization is crucial for the non-convex formulation \eqref{prob:MR}. 
	For our purpose, we adopt the spectral initialization.  
	\blue{To motivate this initialization, 
		note that under uniform sampling, the population least square solution for \eqref{prob:MR} is}:
	\begin{align*}
		d^*\cdot\EE y_i\Y_i = \V^*\C_2^*(\W^{*\top}\otimes \U^{*\top}\U_s) = \M^*(\I\otimes \R^\top\U^{*\top}\U_s),
	\end{align*}
	While this quantity generally differs from $\M^*$, the discrepancy is controlled when $\delta$ is sufficiently small. 
	We therefore take the best rank-$r$ approximation of $\frac{d^*}{n}\sum_{i=1}^ny_i\Y_i$ as our initialization. 
	Crucially, the regressors $\Y_i = \E_{i,2}(\I_{d_3}\otimes \U_s)$ defined in \eqref{def:y} are constructed using the side information, which reduces the ambient dimension from $d_2 \times d_1 d_3$ to $d_2 \times r_1 d_3$.
	To ensure that the initial estimator is incoherent, we add a truncation step in Algorithm \ref{alg:trunc}.
	This algorithm first rescales rows with large norms to satisfy the incoherence condition. Since this operation disrupts the column orthonormality of the singular vectors, the algorithm includes a final re-orthogonalization step. This yields a column-orthonormal matrix $\tilde{\U}$ that remains close to the original subspace while preserving the established incoherence.
	
	\begin{algorithm}[H]
		\caption{Initialization with side information}
		\begin{algorithmic}\label{alg:init}
			\STATE{$\tilde\M_0 = \svd_r ({\frac{d^*}{n}\sum_{i=1}^ny_i\Y_i})$, with its SVD $\tilde\M_0 = \tilde\L_0\tilde\S_0\tilde\R_0^\top$}
			\STATE{$\L_0 = \trunc(\tilde\L_0,\mu)$}
			\STATE{$\R_0 = \trunc(\tilde\R_0,\mu)$}
			\STATE{\textbf{Output:} $\M_0 = \L_0\L_0^\top\tilde\M_0\R_0\R_0^\top$}
		\end{algorithmic}
	\end{algorithm}
	
	\begin{algorithm}[H]
		\caption{\trunc}
		\begin{algorithmic}\label{alg:trunc}
			\STATE{\textbf{Input:} $\U\in\RR^{d\times r}$, $\mu_0>0$}
			\FOR{$i=1,\cdots, d$}
			\STATE{$\U_0(i,:) = \frac{\U(i,:)}{\ltwo{\U(i,:)}}\cdot\min\bigg\{\ltwo{\U(i,:)}, \sqrt{\frac{\mu_0 r}{d}}\bigg\}$}
			\ENDFOR
			\STATE{$\tilde\U = \U_0(\U_0^\top\U_0)^{-1/2}$}
			\STATE{\textbf{Output:} $\tilde\U$}
		\end{algorithmic}
	\end{algorithm}
	
	\section{Theoretical Guarantees}\label{sec:theory}
	In this section, we show how the side information helps reduce the sample complexity in both initialization and convergence phases. 
	We assume the noise $\epsilon_i$ are subgaussian.
	\begin{assumption}\label{assump:noise}
		The noise $\epsilon_i, i=1,\cdots, n$ are i.i.d. and there exists $\sigma>0$, such that for all $s>0$, 
		\begin{align*}
			\EE \exp(s\epsilon_i) \leq \exp(s^2\sigma^2/2). 
		\end{align*}
	\end{assumption}
	Under this assumption, our main theorem goes as follows:
	\begin{theorem}\label{thm:main}
		Suppose Assumption \ref{assump:noise} holds. Assume the following conditions hold:
		\begin{enumerate}[label={(\arabic*)}]
			\item Sample size: $n \geq C_1(d_2\vee d_3r_1)(r_1^3r^4\vee rr^*) \mu_s\mu^4\kappa^8\log d$,
			\item Accuracy of side information: $\delta\leq \min\{\frac{1}{2}\kappa^{-2}, c_1r^{-1/2}\kappa^{-1}\}$,
			\item Signal-to-noise ratio: $\frac{\lambda_{\min}}{\sigma}\geq C_2\sqrt{\frac{d^*(d_2\vee d_3 r_1)r\log d}{n}}\kappa$
		\end{enumerate}
		for some absolute constants $C_1,C_2,C_3,c_1>0$. If we choose $l_{\max}\geq \lceil\log(C_3r\kappa)\rceil$, then the output of our algorithm $\hat\bcalT$ satisfies:
		\begin{align*}
			\frac{1}{d^*}\fro{\hat\bcalT - \bcalT^*}^2 
			&\lesssim \frac{r}{d^*}\lambda_{\max}^2\delta^2 + \frac{(d_2\vee d_3r_1)r}{n}\log d\cdot\sigma^2\\
			&\quad + \frac{ rr^*\mu_s\mu^2}{n^2d_1}(\gamma^2\delta^2 \wedge \mu_s r_1)\lambda_{\max}^2\log^2 d
		\end{align*}
		with probability exceeding $1-12d^{-10}$. 
	\end{theorem}
	The proof of this theorem is composed of two parts: the initialization (Theorem \ref{thm:init}) and the local convergence (Theorem \ref{thm:convergence}). 
	Before we state these two parts as two separate theorems in the following sections, we interpret the meaning of Theorem \ref{thm:main} from four perspectives: sample complexity, statistical error, signal-to-noise ratio condition, and the conditions on side information. 
	
	\vspace{0.2cm}
	
	\noindent\textbf{\emph{Sample-size comparison under subspace information. }}
	For bounded ranks, incoherence, and condition number, our theorem requires $O_{\mu,\r,\kappa}\big((d_2\vee d_3)\log d\big)$ samples for both initialization and local convergence. Under the additional subspace-information assumption, this has smaller ambient-dimensional dependence than the passive uniform-sampling Tucker-completion guarantees summarized in Table \ref{table:literature}, which scale as $O_{\mu,\r,\kappa}\big((\sqrt{d_1d_2d_3}+d)\log^b d\big)$ for $b\geq 2$ under their respective assumptions. This is not a comparison under identical information: most listed baselines do not assume side information, whereas \approach does.
	The bound is independent of $d_1$. For illustration, let $d_2=d_3$. Among the listed passive uniform-sampling guarantees, the ambient-dimensional terms scale as $\tilde O_{\mu,\r,\kappa}(d_1)$ when $d_1\geq d_2^2$ and as $\tilde O_{\mu,\r,\kappa}(\sqrt{d_1d_2d_3})$ when $d_1\leq d_2^2$ and the square-root term dominates, whereas the \approach bound scales as $\tilde O_{\mu,\r,\kappa}(d_2)$ under accurate side information. Thus, this comparison quantifies the potential gain supplied by the side information. See also Figure \ref{fig:sample_compare}.
	This gain does not extend uniformly to rank dependence. Assuming $d_1=d_2=d_3=d$ and $r_1=r_2=r_3=r$, our theorem requires $\tilde O_{\mu,\kappa}(dr^8)$ samples, whereas \cite{xia2021statistically} gives $\tilde O_{\mu,\kappa}(d^{3/2}r^4)$. Accordingly, \approach is most favorable when the ambient dimensions are sufficiently large relative to the rank. Improving the rank dependence is left for future work.
	
	\begin{figure}[H]
		\centering
		\includegraphics[width=0.5\textwidth]{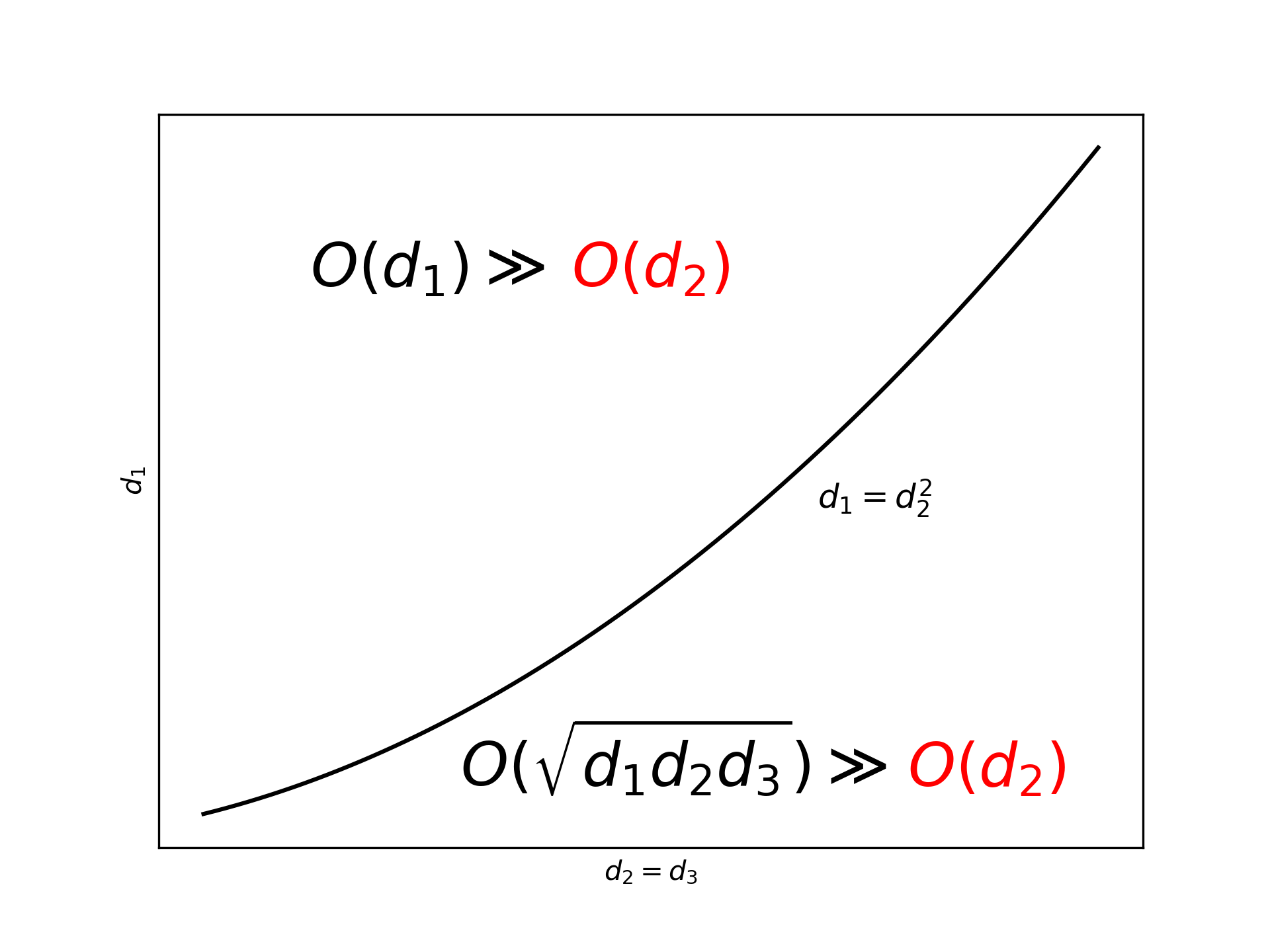}
		\caption{Schematic comparison of the ambient-dimensional sample-size terms discussed above under the stated information assumptions. The \approach term is highlighted in red, while the comparator terms are shown in black.}
		\label{fig:sample_compare}
	\end{figure}

	For additional context, we compare with non-Tucker formats, noting that CP rank-1, tensor-train rank-$(1,1)$, and Tucker rank-$(1,1,1)$ describe the same class. In the balanced, fixed-rank setting, the cited passive uniform-sampling CP and tensor-train guarantees scale as $O_{\mu,\kappa}(d^{3/2}\log^b d)$ for $b\geq 2$, whereas the \approach bound scales as $O_{\mu,\kappa}(d\log d)$ under its subspace-information assumption; see Table \ref{table:literature:nonTucker}.
	The covariate-assisted COSTCO analysis \citep{ibriga2022covariate} has a $\widetilde O_{\mu,\kappa}(d^{3/2})$ sample requirement even when a suitable initialization is available. COSTCO jointly factors the tensor and covariate information, whereas \approach first uses the estimated subspace to reduce the problem to matrix regression. This methodological difference yields a different sample-size bound under the subspace-accuracy assumptions of \approach, but it does not imply general computational or runtime dominance.
	
	\vspace{0.2cm}
	
	\noindent\textbf{\emph{Resulting statistical error. }}
	In general, when the side information is accurate under $\twoinf{\cdot}$ norm, that is when $\gamma \leq \frac{n}{\sqrt{d_2d_3r^*}\log d}\frac{1}{\sqrt{\mu_s}\mu},$
	the final statistical error is $$\frac{1}{d^*}\fro{\hat\bcalT - \bcalT^*}^2\lesssim  \frac{(d_2\vee d_3r_1)r}{n}\log d\cdot\sigma^2 + \frac{r}{d^*}\lambda_{\max}^2\delta^2.$$ 
	Under the sample size condition in Theorem \ref{thm:convergence}, this condition reduces to $\gamma\lesssim 1$. Further discussion is provided below. If the side information is also accurate under $\fro{\cdot}$, namely, $\delta \lesssim  \frac{\sigma}{\lambda_{\max}}\sqrt{\frac{d^*(d_2\vee d_3r_1)}{n}\log d},$ the final statistical error reduces to
	\begin{align*}
		\frac{1}{d^*}\fro{\hat\bcalT - \bcalT^*}^2\lesssim  \frac{(d_2\vee d_3r_1)r}{n}\log d\cdot\sigma^2. 
	\end{align*}
	For comparison, \cite{xia2021statistically} gives the following bound for a power-iteration estimator $\hat\bcalT_{\text{PI}}$:
	$$\frac{1}{d^*}\fro{\hat\bcalT_{\text{PI}} - \bcalT^*}^2\lesssim\frac{d\max\{r_1,r_2,r_3\}}{n} \log d\cdot\sigma^2+\frac{d\max\{r_1,r_2,r_3\}}{n} \log d\cdot\frac{\lambda_{\max}^2}{d^*}. $$
	Relative to this bound, the \approach bound does not contain the second signal-strength term and has a smaller ambient-dimensional factor when $d_1$ is much larger than $d_2,d_3$, subject to sufficiently accurate side information.
	As a reference, we also compare with the COSTCO estimator (Covariate-Assisted Sparse Tensor Completion) proposed in \cite{ibriga2022covariate}. To align our comparison with their low CP rank framework, we restrict our analysis to the Tucker rank-$(1, 1, 1)$ case (equivalent to CP-rank 1).
	While their results exhibit notable improvement in estimation error for the mode coupled with side information, the overall statistical error of the tensor is influenced by the largest estimation error across all modes, which is  
	$$\frac{1}{d^*}\fro{\hat\bcalT_{\text{COSTCO}} - \bcalT^*}^2\lesssim\frac{\kappa^2d\log d}{n}\sigma^2.$$ 
	Accordingly, in this rank-one comparison and under reliable side information, our upper bound is smaller when $d_1\gg d_2,d_3$ and has the same ambient-dimensional order in the balanced case.
	
	\vspace{0.2cm}
	
	\noindent\textbf{\emph{Weaker signal-to-noise ratio condition. }}
	Aside from the sample complexity, the signal-to-noise ratio is also an important condition in noisy tensor completion. 
	If the noise level is too large relative to the signal, reliable recovery is not possible.
	It is assumed in noisy tensor completion \citep{xia2021statistically} that $\lambda_{\min}/\sigma\gtrsim_{\r} \sqrt{\frac{(d^*)^{3/2}}{n}\log^{5}d} + \sqrt{\frac{d^* d}{n}\log^{10} d}$, where the notation $\gtrsim_{\r}$ hides a constant factor that depends only on the rank $\r$. 
	The \approach condition does not contain the first term when $\kappa=O(1)$ and has the same order as the condition in \cite{ibriga2022covariate}, under the respective side-information assumptions.
	
	\vspace{0.2cm}
	
	\noindent\textbf{\emph{Remarks on conditions concerning side information. }}
	The accuracy of side information plays an important role in our analysis. 
	To reach the noise-dominated error rate in the previous discussion, we require
	\begin{align*}
		\gamma \lesssim 1, \quad \delta \lesssim \frac{\sigma}{\lambda_{\max}}\sqrt{\frac{d^*(d_2\vee d_3r_1)}{n}\log d}. 
	\end{align*}
	The condition $\gamma\lesssim 1$ means the difference between $\U_s$ and $\U^*$ is spread out evenly across each row. 
	Larger sample size, or larger signal-to-noise ratio will lead to a stricter condition on $\delta$. This is because the error $\U_s$ brings in is non-vanishing and is independent of the original tensor completion problem. 
	The following examples illustrate regimes in which these conditions can hold.
	
	\noindent\textit{Example 1: Fully revealed matrix.}
	In this example, we consider the side information of matrix form $\M_s = \U^*\bSigma_s\V_s^\top + \E_s\in\RR^{d_1\times d_s}$.  
	This side information is also considered in \cite{ibriga2022covariate}.  For simplicity, we assume $\E_s$ has i.i.d. $N(0,\sigma_s^2)$ entries. 
	We may take $\U_s$ to be the top $r_1$ left singular vectors of $\M_s$. Then using Wedin's sin$\Theta$ Theorem, we can obtain with high probability 
	\begin{align*}
		\delta \lesssim \frac{\sqrt{(d_1\vee d_s)\log d}\cdot \sigma_s}{\lambda_{\min,s}},
	\end{align*}
	where $\lambda_{\min,s} = \sigma_{r_1}(\bSigma_s)$ is the signal strength. For illustration purpose, we assume $d_1\asymp d_2\asymp d_3\asymp d_s$, and $\sigma_s/\lambda_{\min,s}\asymp \sigma/ \lambda_{\min}$, $\r,\kappa\asymp 1$. Then the condition on $\delta$ is met as long as $n\lesssim d^*$. 
	This condition is compatible with the sample-poor regimes in which side information is most useful. When the number of observations is comparable to $d^*$, the need for auxiliary information is naturally reduced.
	The other condition $\gamma\lesssim 1$ is more involved to verify. Fortunately, this result was established in various existing literature (e.g. \cite[Theorem 3.1]{cape2019two} and \cite[Section 4]{chen2021spectral}). 
	
	\noindent\textit{Example 2: Incomplete matrix.}
	Next we consider the incomplete case, where we only observe $n_s$ entries uniformly at random in the matrix $\M_s = \U^*\bSigma_s\V_s^\top + \E_s\in\RR^{d_1\times d_s}$.
	We shall invoke the existing literature on matrix completion (e.g., \cite{chen2019noisy}), where the bound provided indicates 
	\begin{align*}
		\delta \lesssim \sqrt{\frac{d_1d_s}{n_s}}\frac{\sqrt{(d_1\vee d_s)\log d}\cdot \sigma_s}{\lambda_{\min,s}}.
	\end{align*}
	Assume $d_1\asymp d_2\asymp d_3\asymp d_s$, and $\sigma_s/\lambda_{\min,s}\asymp \sigma/ \lambda_{\min}$, $\r,\kappa\asymp 1$, then the condition on $\delta$ reduces to 
	$\frac{n}{n_s}\lesssim d$.
	The minimal samples required is thus $n_s = n = O(d\log d)$.
	Here $n_s = \Omega(d \log d)$ is independently required to guarantee the validity of the subspace estimation from the incomplete matrix side information \cite{chen2019noisy}. 
	And the infinity norm error provided in \cite{chen2019noisy} shows the error is evenly spread out in each column of the subspace, indicating $\gamma\lesssim 1$. 

	\subsection{Initialization}
	For non-convex problem, it is of key importance to design the initialization to make sure it is near the groundtruth. 
	Existing initialization algorithms such as spectral initialization \citep{cai2019nonconvex,tong2022scaling}, second-order moment method \citep{xia2019polynomial} can indeed provide a desired initial estimator within polynomial time but requires $\tilde O_{\mu,\r,\kappa}(\sqrt{d_1d_2d_3} + d)$ samples, where $d = \max\{d_1,d_2,d_3\}$. 
	
	The COSTCO method of Ibriga and Sun \cite{ibriga2022covariate} initializes the shared component using the covariate matrix and the non-shared components using the robust tensor power method \cite{anandkumar2014guaranteed}. Our side-information-aided spectral initialization is instead designed for the reduced matrix-regression formulation and is covered by our theoretical analysis.
	We summarize the result in the next theorem.

	\begin{theorem}\label{thm:init}
		Suppose Assumption \ref{assump:noise} is satisfied. And the following conditions hold:
		\begin{enumerate}[label={(\arabic*)}]
			\item Sample size: $n\gtrsim (d_3r_1\vee d_2)rr^*\log d\cdot \kappa^4\big(\mu^3\vee(\mu_s\mu^3)^{1/2}\big)$;
			\item Accuracy of side information: $\delta \leq \frac{1}{2}\kappa^{-2}$;
			\item Signal-to-noise ratio: $\frac{\lambda_{\min}}{\sigma}\gtrsim\sqrt{\frac{d^*r(d_2\vee d_3r_1)\log d}{n}}\kappa $. 
		\end{enumerate}
		Then with probability exceeding $1-2d^{-10}$, 
		the output of Algorithm \ref{alg:init} satisfies $\fro{\M_0-\M^*}\leq c_1\lambda_{\min}$ for some small absolute constant $c_1 >0$, and $\incoh(\M_0)\leq3\mu$. 
	\end{theorem}

	\subsection{Local Convergence}
	
	Once we obtain an initialization $\M_0$ in a small neighborhood of $\M^*$, Theorem \ref{thm:convergence} shows that Algorithm \ref{alg:rgrad} converges linearly up to a non-vanishing error floor induced by the side information and observation noise. Despite this additional error, accurate side information can reduce the sufficient sample-size requirement. The recovered tensor is then obtained by combining the final matrix estimate with the estimated subspace $\U_s$.
	
	\begin{theorem}\label{thm:convergence}
		Suppose Assumption \ref{assump:noise} holds. 
		The initialization $\M_0$ satisfies $\fro{\M_0 -\M^*}\leq c_0\lambda_{\min}$ for some small absolute constant $c_0>0$ and $\incoh(\M_0)\leq3\mu$. 
		Assume the following conditions hold:
		\begin{enumerate}[label={(\arabic*)}]
			\item Sample size: $n \geq C_1(d_2\vee d_3r_1)r_1^3r^4 \mu_s\mu^4\kappa^8\log d$,
			\item Accuracy of side information: $\delta\leq c_1r^{-1/2}\kappa^{-1}$,
			\item Signal-to-noise ratio: $\frac{\lambda_{\min}}{\sigma}\geq C_2\sqrt{\frac{d^*(d_2\vee d_3 r_1)r}{n}\log d}$
		\end{enumerate}
		for some absolute constants $C_1,C_2,c_1>0$.
		If we choose $l_{\max}\geq \lceil\log(C_3r\kappa)\rceil$,
		we have that the output of Algorithm \ref{alg:rgrad} satisfies 
		\begin{align*}
			\frac{1}{d^*}\fro{\hat\bcalT - \bcalT^*}^2 
			&\lesssim \frac{r}{d^*}\lambda_{\max}^2\delta^2 + \frac{(d_2\vee d_3r_1)r}{n}\log d\cdot\sigma^2\\
			&\quad + \frac{ rr^*\mu_s\mu^2}{n^2d_1}(\gamma^2\delta^2 \wedge \mu_s r_1)\lambda_{\max}^2\log^2 d
		\end{align*}
		with probability exceeding $1-10d^{-10}$. 
		
	\end{theorem}

	\section{Simulations}\label{sec:simulation}
	In this section we perform numerical experiments to verify our theoretical results. 
	We compare our algorithm \approach mainly with COSTCO \citep{ibriga2022covariate} that incorporates the side information in a different way, and standard tensor completion algorithm using Riemannian gradient descent \citep{cai2022provable}, which we shall refer to TC in the following.  
	
	\subsection{Statistical Error}
	In the first experiment, we inspect the performance of our algorithm when noise level, or the accuracy of the side information changes.
	We generate a random tensor $\bcalT^* \in \RR^{d\times d\times d}$ with $d=30$ and Tucker rank $\r= (2,2,2)$ by first constructing a tensor from factors and a core with i.i.d. $\text{Unif}([0,1])$ entries, and then performing HOSVD to obtain the standardized decomposition $\bcalT^*= \bcalC^*\times_1\U^*\times_2\V^*\times_3\W^*$. 
	This low rank tensor $\bcalT^*$ has its incoherence bounded by 3 and $\lambda_{\max}\approx 50, \lambda_{\min}\approx 10$. 
	The side information is provided in the form of a matrix $\M_s = \U^*\C + \E_s \in\RR^{d \times n_s}$, whose left singular vectors are coupled with the first mode of the tensor. 
	Here $\C\in\RR^{r_1\times n_s}$ represents the coefficients having i.i.d. $N(0,1)$ entries, and $\E_s$ is the noise having i.i.d. $N(0,\sigma_s^2)$ entries. 
	We extract the subspace $\U_s = \svd_{r}(\M_s)$. The accuracy of side information measured in $\delta = \fro{\U_s\U_s^\top - \U^*\U^{*\top}}$.
	We denote the output of our algorithm by $\hat\bcalT$. 

	We fix the sample size $n = 1000$. For the side information, we fix $n_s = 1000$, and we vary $\sigma_s\in\{0.1,0.3,0.5\}$. 
	The resulting $\delta$ takes the value within $\{0.037, 0.115, 0.204\}$ respectively.
	We also vary the noise level 
	$$\sigma\in\{0, 10^{-4}, 5\cdot 10^{-4}, 10^{-3}, 5\cdot 10^{-3}, 0.01,0.02,\cdots,0.1\}.$$
	For each $(\delta, \sigma)$ pair, we conduct 30 independent trials and record the final error $\fro{\hat \bcalT - \bcalT^*}$. 
	Figure \ref{fig:change_delta} illustrates the result. 
	As we can see from this figure, for each given $\delta$, there is a flat region when the noise is small. 
	Even in the case when the noise is zero, the final statistical error $\fro{\hat \bcalT - \bcalT^*}$ does not vanish. This is due to the bias of $\U_s$, and is consistent with our theory. 
	And the quality of side information greatly influences the final error rate if the noise is small, but has a weaker impact as noise gets larger, which also aligns with our theory. 
	Finally, the plot indicates that the final error scales linearly with the noise level.
	In terms of computational efficiency, we further observed that the runtime remained robust to variations in noise and side information quality, averaging approximately 0.69 seconds per trial.
	
	\begin{figure}[htbp]
		\begin{minipage}{0.48\textwidth}
			\centering
			\includegraphics[width=\linewidth]{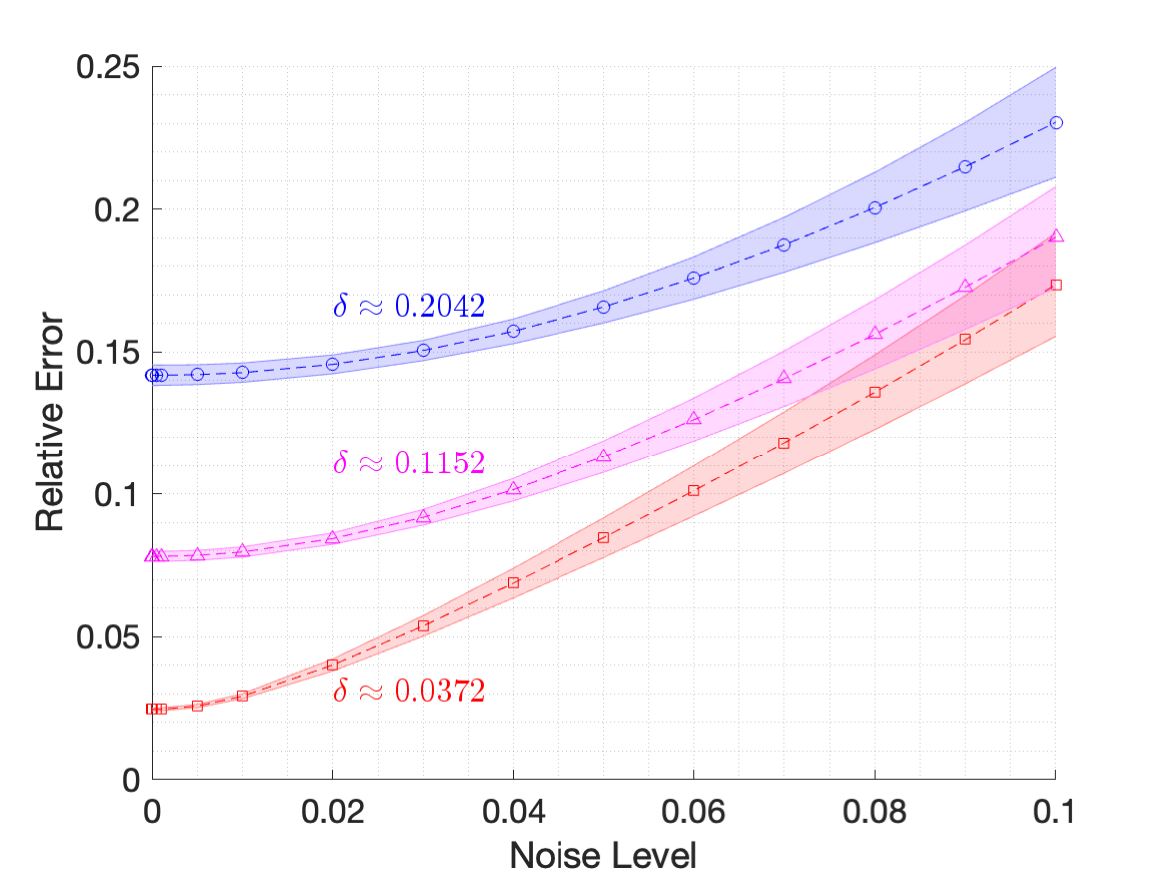}
			\caption{Performance with respect to the change of $\delta$ and $\sigma$. }
			\label{fig:change_delta}
		\end{minipage}
		\hfill
		\begin{minipage}{0.48\textwidth}
			\centering
			\includegraphics[width=\linewidth]{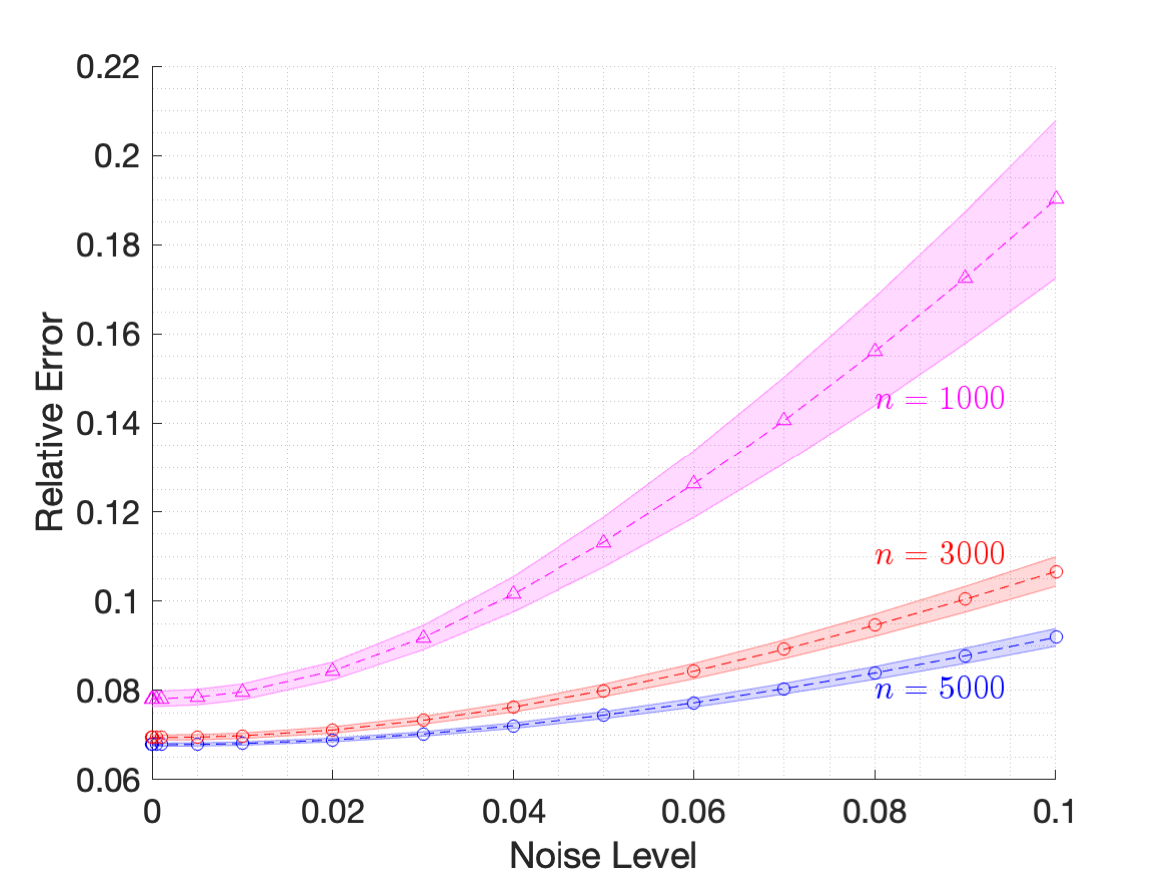}
			\caption{Performance with respect to the change of $n$ and $\sigma$. }
			\label{fig:change_n}
		\end{minipage}
	\end{figure}
	
	In the second experiment, we fix the accuracy of the subspace $\delta \approx 0.115$, and we alter the sample size $n\in\{1000,3000,5000\}$, the noise level 
	$$\sigma\in\{0, 10^{-4}, 5\cdot 10^{-4}, 10^{-3}, 5\cdot 10^{-3}, 0.01,0.02,\cdots,0.1\}.$$
	For each $(\delta, \sigma)$ pair, we conduct 30 independent trials and record the final error $\fro{\hat \bcalT - \bcalT^*}$. 
	Figure \ref{fig:change_n} displays the final result. We can conclude from this figure when the sample size gets larger, the final error rate becomes smaller. Moreover, the width of the error becomes narrower as sample size increase, indicating the reduction in the variance of $\fro{\hat \bcalT - \bcalT^*}$. 

	\subsection{Phase Transition}
	In this experiment, we test the recovery ability of the proposed algorithm in the framework of phase transition, which compares the number of measurements $n$, and a cubic tensor of size $d\times d\times d$ of Tucker rank $\r = (2,2,2)$. 
	For the phase transition, the side information is generated in the following way: $\M_s = \U^*\C + \E_s$, where $\U^*\in\RR^{d\times r}$ is the subspace along the first mode of $\bcalT^*$, $\C\in\RR^{r\times n_s}$ having i.i.d. $N(0,1)$ entries is the coefficient, and $\E_s$ is the Gaussian noise having i.i.d. $N(0,\sigma_s^2)$ entries. In this experiment, we set $n_s = 1000$, and $\sigma_s = 0.01$. 
	For each $(d,n)$ pair, we conduct 30 random instances, and a test is considered to be successful if $\fro{\hat\bcalT - \bcalT^*}/\fro{\bcalT^*}<0.01$. 
	To compare the empirical sampling behavior of the methods, we also present phase-transition plots for TC and COSTCO.
	We use $\M_s$ as the side information for COSTCO. 
	For these two algorithms, we vary $d\in\{20,25,\cdots,50\}$, and the sample complexity is adjusted so that we can see the phase transition line. The results are displayed in Figure \ref{fig:phasetransition_TC} and \ref{fig:phasetransition_COSTCO}.
	For our algorithm, we vary $d\in\{20,30,\cdots, 100\}$, and the sample complexity $n\in\{500, 1000,\cdots, 4000\}$. The result is displayed in Figure \ref{fig:phasetransition}. 
	
	\begin{figure}[H]
		\begin{minipage}{0.48\textwidth}
			\centering
			\includegraphics[width=\linewidth]{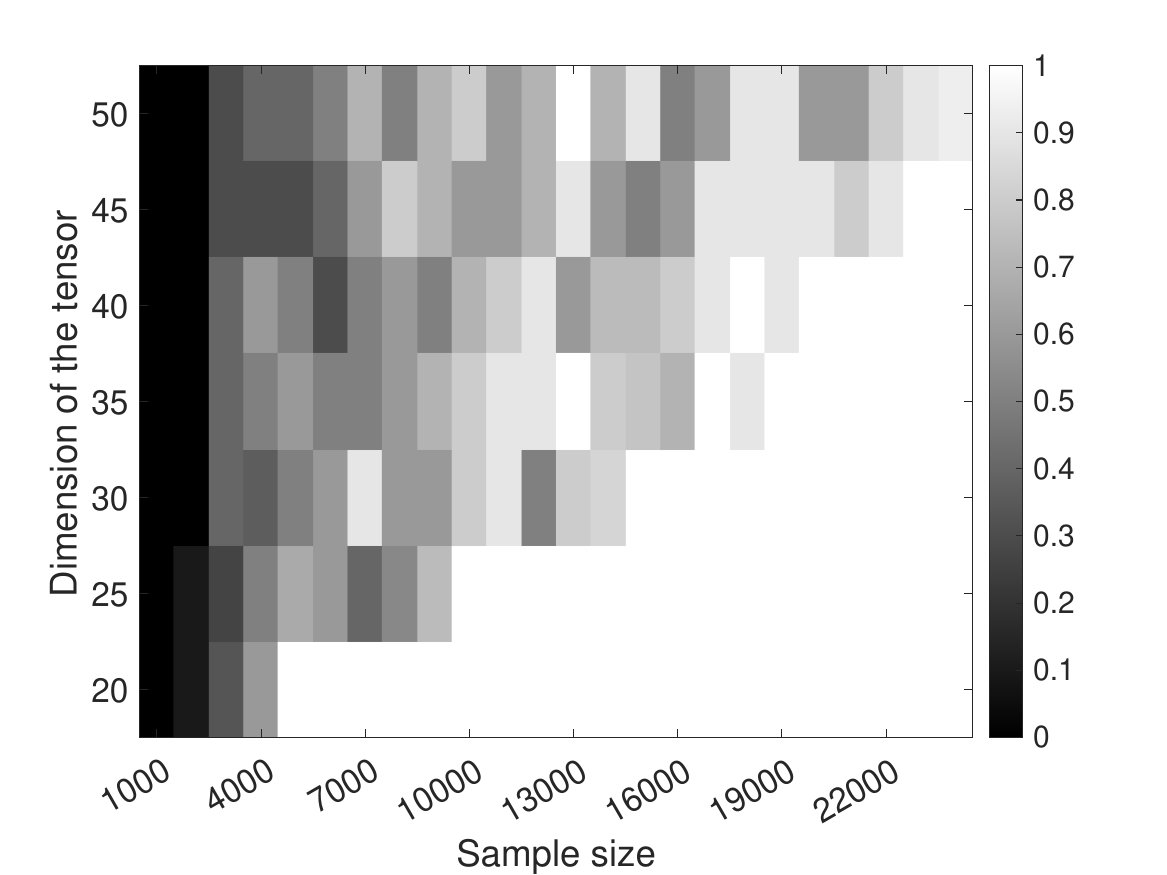}
			\caption{Phase transition of TC}
			\label{fig:phasetransition_TC}
		\end{minipage}
		\hfill
		\begin{minipage}{0.48\textwidth}
			\centering
			\includegraphics[width=\linewidth]{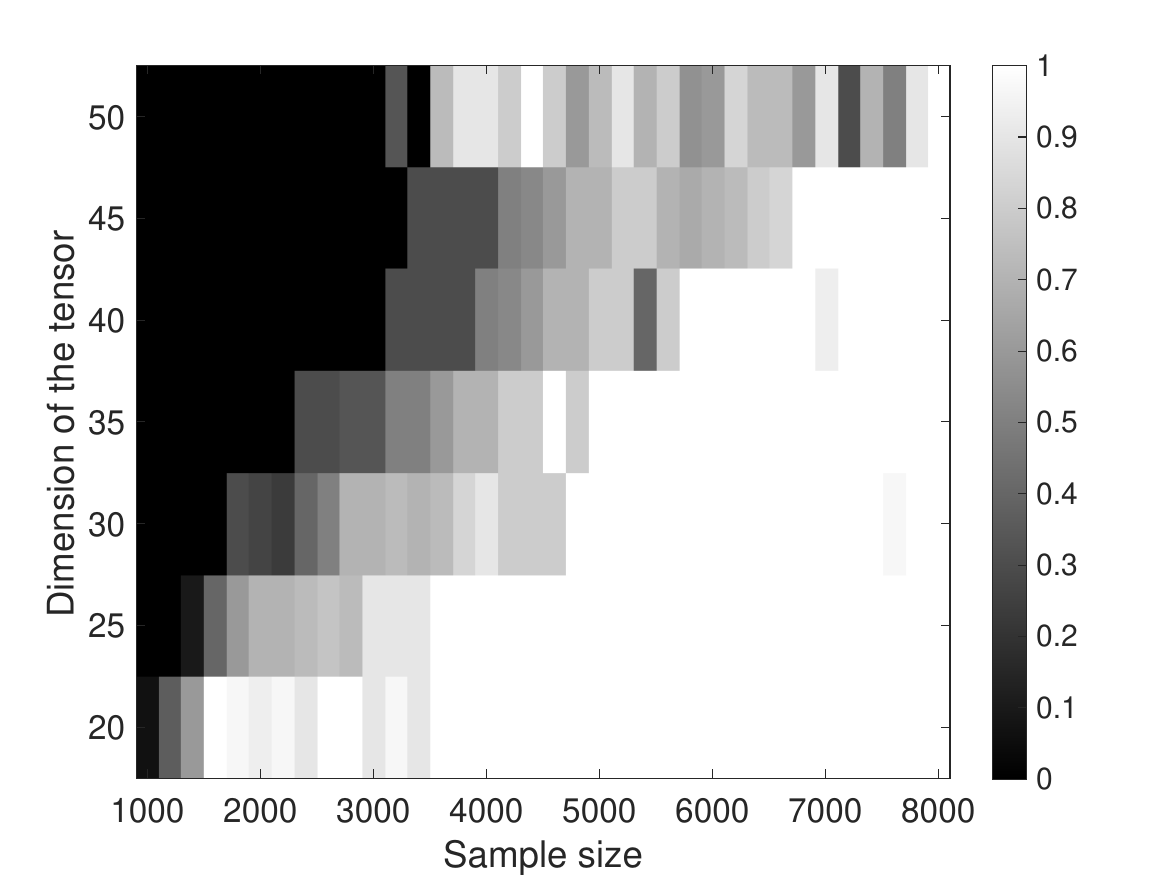}
			\caption{Phase transition of COSTCO }
			\label{fig:phasetransition_COSTCO}
		\end{minipage}
	\end{figure}

	\begin{figure}[H]
		\centering
		\includegraphics[width=0.5\textwidth]{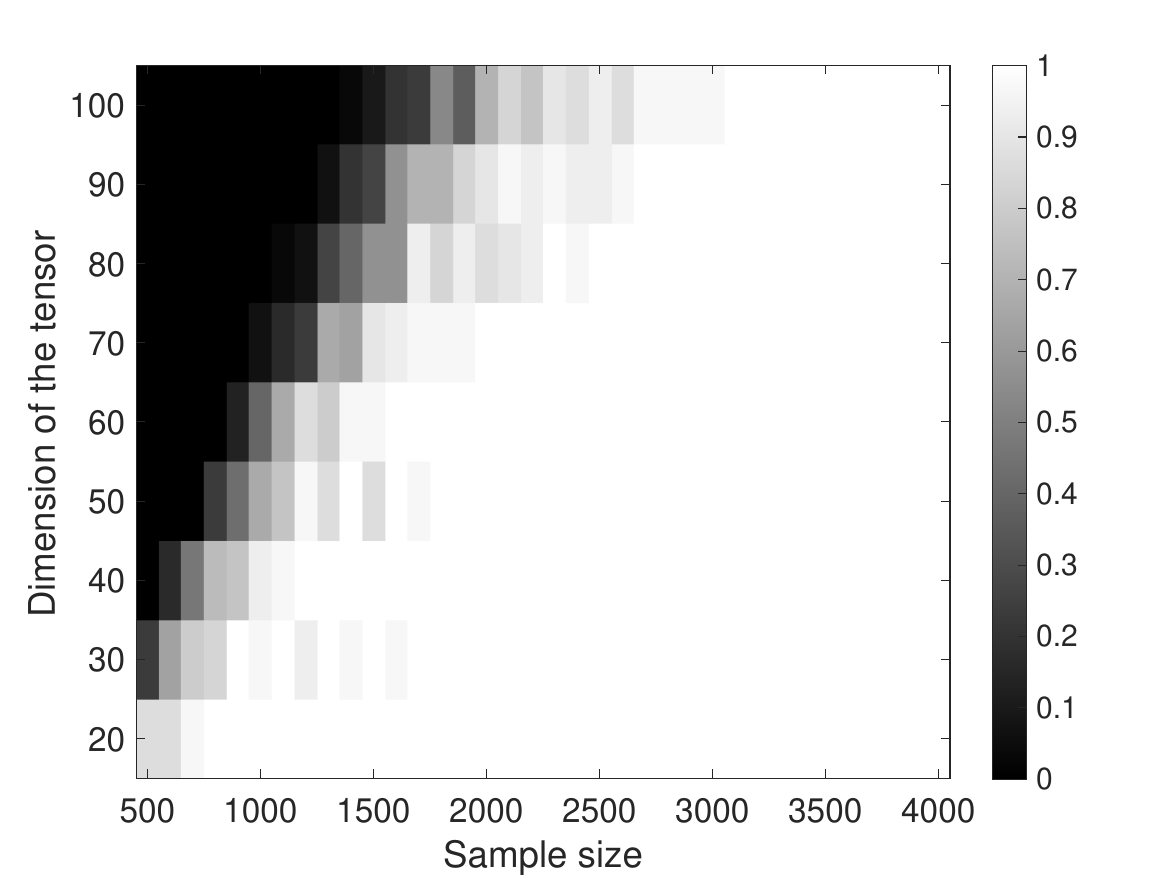}
		\caption{Phase transition of \approach}
		\label{fig:phasetransition}
	\end{figure}
	
	In these experiments, TC requires the most observations for successful recovery, COSTCO requires fewer, and \approach requires the fewest among the three tested methods. The approximately linear phase-transition boundary for \approach is consistent with the dimensional dependence predicted by our theory.

	\section{Real-World Application: Global Ionospheric TEC Tensor Completion}\label{sec:realdata}

	In this section, we apply the proposed \approach methodology to a realistic problem in space weather forecasting: the reconstruction of global Total Electron Content (TEC) maps. Global TEC maps are critical for calibrating global navigation satellite systems (GNSS). However, since TEC data is primarily collected by land-based GNSS receivers, the raw observations are sparse and non-uniform \cite{rideout2006automated}, featuring massive continuous data gaps over oceans and remote regions\footnote{The raw TEC dataset utilized in these experiments is publicly available at (\url{https://deepblue.lib.umich.edu/data/concern/data_sets/nc580n00z}).}. 

	\vspace{0.2cm}
	
	\noindent\textbf{\emph{Data-Adaptive Safeguard via Sample Splitting.}} 
	Because the theoretical subspace distance $\delta$ is unknown in practice, we propose a data-adaptive, sample-splitting safeguard to prevent the incorporation of misleading side information. We randomly partition the observed entries $\Omega_{\text{obs}}$ into a training set $\Omega_{\text{train}}$ and a validation set $\Omega_{\text{val}}$. Both the proposed \approach and a baseline estimator (e.g., vanilla Riemannian gradient descent) are computed exclusively on $\Omega_{\text{train}}$. We then evaluate their prediction errors on the held-out $\Omega_{\text{val}}$. We accept the side information only if \approach achieves a lower error or if the baseline fails to converge due to extreme sparsity. Otherwise, we revert to the baseline, ensuring the prior is utilized solely when it empirically enhances generalization on unseen data.

	To comprehensively evaluate our method, we conduct two sets of experiments: (1) a semi-synthetic simulation under controlled uniform sparsity, and (2) a real data completion task under operational structured missingness. 
	Throughout these experiments, we benchmark our algorithm against two existing methods: traditional tensor completion (TC) via Riemannian gradient descent, which does not utilize auxiliary data, and COSTCO \cite{ibriga2022covariate}, a coupled tensor completion algorithm that incorporates side information.

	\subsection{Motivation and Physically Grounded Side Information}
	
	We aim to reconstruct the global TEC distribution for a specific target day (December 25, 2019). After preprocessing, the target data is formulated as a third order tensor $\bcalT^* \in \RR^{45 \times 90 \times 72}$, where the three modes correspond to latitude ($d_1 = 45$), longitude ($d_2 = 90$), and time of day ($d_3 = 72$, representing measurements recorded at 20-minute intervals). 

	To extract a physically meaningful spatial subspace, we derive our side information exclusively from historical data spanning the preceding week (December 18 to December 24, 2019). Specifically, we compute the temporal average of the 7-day historical continuous TEC maps pre-imputed by the well-established VISTA algorithm \citep{sun2023complete} to obtain a two-dimensional aggregated spatial matrix.
	From a geophysical perspective, this averaging effectively filters out transient daily ionospheric fluctuations (e.g., short-term irregularities and local geomagnetic disturbances) while preserving the highly stable, macroscopic spatial background governed by the Earth's geomagnetic field. To explicitly extract the subspace information along the second direction (longitude), we apply SVD directly to this averaged matrix. By examining the scree plot of the singular values, we truncate the subspace at rank $10$. This procedure yields the spatial subspace $\U_s \in \RR^{90 \times 10}$, providing a physically grounded prior for the target day.
	
	\subsection{Semi-Synthetic Evaluation under Controlled Uniform Sparsity}
	To compare empirical performance under controlled sparsity, we design a semi-synthetic experiment using the VISTA-imputed TEC map \cite{sun2023complete} as a proxy ground truth. This preserves authentic spatial-temporal ionospheric structures while enabling exact global error calculation. We uniformly sample the target tensor $\bcalT^*$ at $1\%$ and $10\%$ observation ratios and compare \approach with TC via Riemannian gradient descent and the coupled method COSTCO \cite{ibriga2022covariate}.

	Table \ref{table:tec_relative_error} summarizes the quantitative results. At the $1\%$ observation rate, TC diverges and COSTCO has a large relative error in this experiment, whereas \approach attains a relative error of $0.254$. At the $10\%$ observation rate, \approach also has the lowest error among the three methods. These results are consistent with, but do not by themselves establish, the theoretical sample-size comparison.

	\begin{table}[ht]
		\centering
		
		\begin{tabular}{lccc}
			\hline
			{Sampling Rate} & {TC} & {COSTCO} & \textbf{\approach (Ours)} \\
			\hline
			$1\%$  & \text{Diverges} & 8.99 & \textbf{0.254} \\
			$10\%$ & 0.134 & 0.167 & \textbf{0.124} \\
			\hline
		\end{tabular}
		\caption{Relative error of TEC map reconstruction under different uniform sampling rates. The best results are highlighted in \textbf{bold}.}
		\label{table:tec_relative_error}
	\end{table}

	Qualitatively, Figures \ref{fig:tec_slice10} and \ref{fig:tec0_01_10} compare the proxy ground truth and algorithmic reconstructions at the 10th temporal slice. At 1\% sparsity, \approach preserves the spatial topology. Furthermore, at $10\%$, only \approach accurately recovers the localized high-intensity TEC regions (central red areas), whereas the baseline methods fail to capture these peak intensities, yielding heavily blurred and over-smoothed structures.

	\begin{figure}[H]
		\centering
		\includegraphics[width=0.5\textwidth]{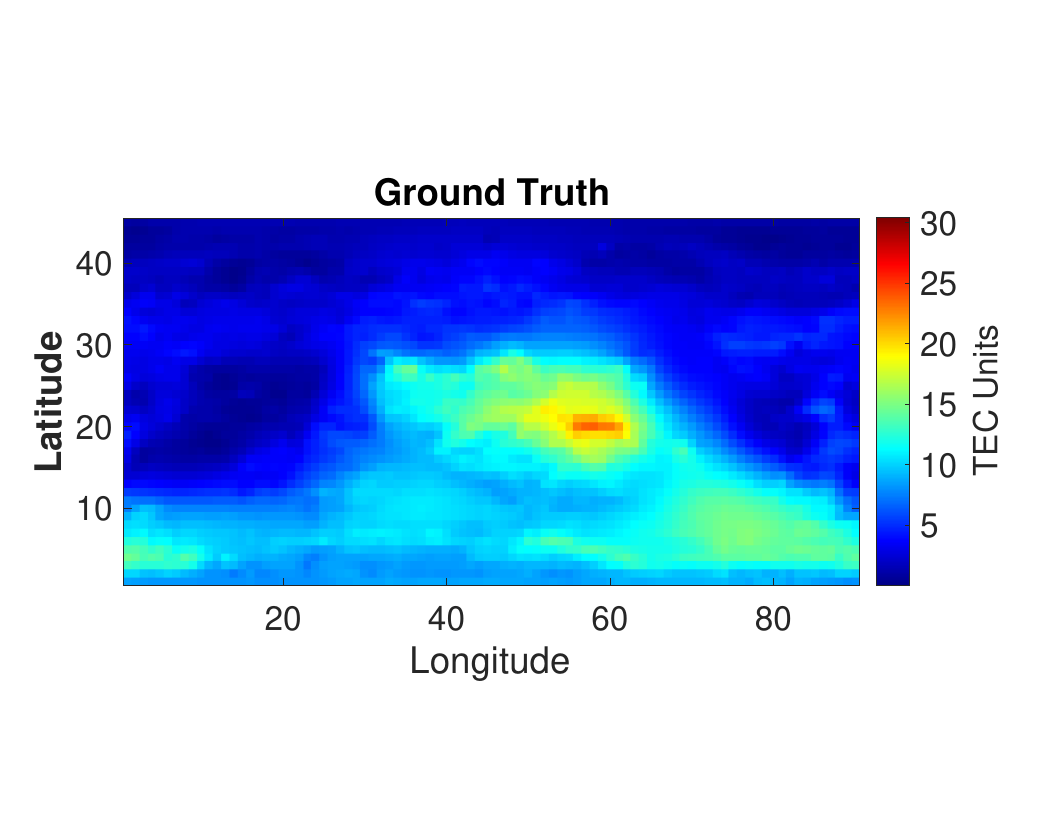}
		\vspace{-1cm}
		\caption{Proxy ground truth of the global TEC map at the 10th temporal snapshot.}
		\label{fig:tec_slice10}
	\end{figure}
	
	\begin{figure}[H]
		\centering
		\includegraphics[width=\textwidth]{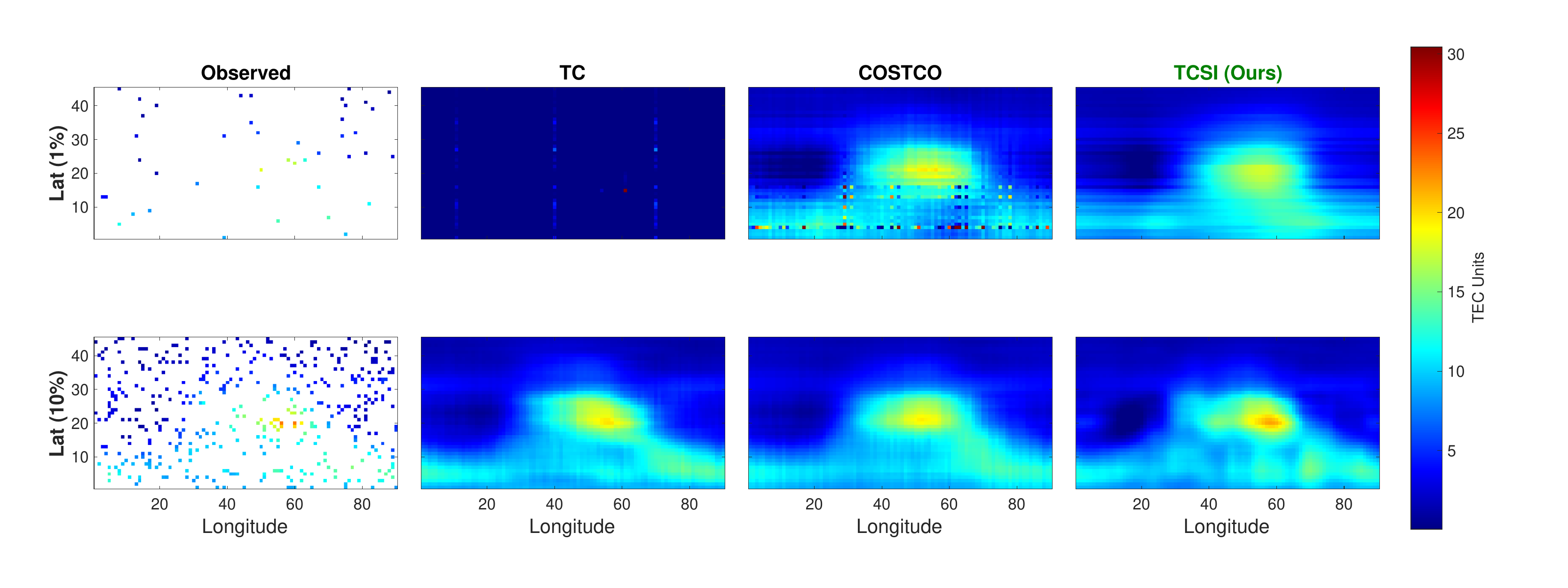}
		
		\caption{Visual comparison of TEC map reconstructions under 1\% and 10\% uniform sampling rates.}
		\label{fig:tec0_01_10}
	\end{figure}

	\subsection{Real Data Completion under Structured Missingness}
	Unlike the uniform sampling in the previous section, real-world GNSS observations suffer from severe structured (block) missingness. After preprocessing, the overall observation rate of the raw TEC data on December 25, 2019, sourced from the Madrigal TEC database \cite{rideout2006automated}, is approximately $16\%$. Crucially, these observations are heavily clustered over landmasses, leaving extensive continuous data gaps over the oceans. In this scenario, the subspace information $\U_s$ serves as a valuable structural prior to help bridge these unobservable regions.

	Because the true global dense tensor is unknown, we evaluate the models strictly on the available non-uniform observations, $\Omega_{\text{obs}}$. To execute our data-adaptive safeguard and objectively assess performance, we randomly partition $\Omega_{\text{obs}}$ into training ($90\%$), validation ($5\%$), and test ($5\%$) sets. We first invoke the safeguard by comparing prediction errors on the validation set: \approach achieves a slightly lower relative error ($0.126$) than the baseline TC ($0.128$). This satisfies the safeguard condition to accept $\U_s$, confirming that the underlying physical spatial basis remains consistent across consecutive days.
	
	Having validated the side information, we evaluate all three algorithms on the strictly held-out test set. As shown in Table \ref{table:tec_real_missing}, \approach achieves the lowest test relative error ($0.117$). Furthermore, the visual comparison in Figure \ref{fig:tec_real_missing} demonstrates that \approach provides a more refined and structurally continuous recovery over the oceanic gaps compared to the baseline methods. 
	
	\begin{table}[ht]
		\centering
		\begin{tabular}{ccc}
			\hline
			{TC} & {COSTCO} & \textbf{\approach (Ours)} \\
			\hline
			0.124 & 0.156 & \textbf{0.117} \\
			\hline
		\end{tabular}
		\caption{Relative error on held-out real TEC observations.}
		\label{table:tec_real_missing}
	\end{table}
	\begin{figure}[H]
		\centering
		\includegraphics[width=\textwidth]{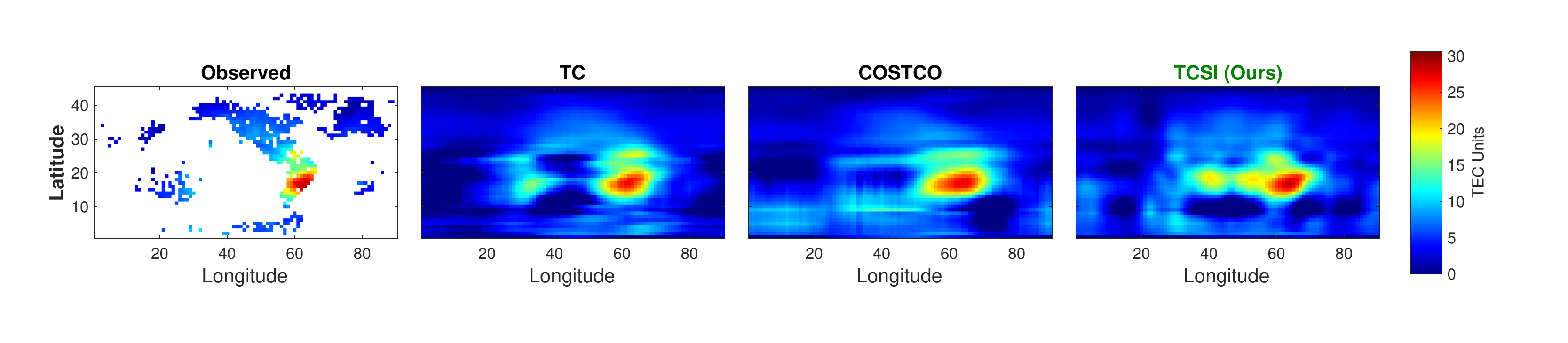}
		\caption{Visual comparison of real-world TEC map reconstructions.}
		\label{fig:tec_real_missing}
	\end{figure}
	
	Overall, these experiments illustrate the practical utility of incorporating side information for this TEC reconstruction task. The spatial subspace $\U_s$ provides a useful structural prior under sparse, block-missing observations, and \approach achieves lower held-out errors than the compared methods on this dataset.

	\bibliographystyle{plain}
	
	\bibliography{Bibliography-MM-MC}

	\clearpage
	\appendix
	\pdfbookmark[0]{Appendices}{appendices}
	\section*{Appendices}
	
	\section{Additional Algorithm}\label{app:alg}
	\begin{algorithm}[H]
		\caption{Higher-order Singular Value Decomposition (HOSVD)}
		\begin{algorithmic}\label{alg:hosvd}
			\STATE{\textbf{Input: }$\bcalT\in\RR^{d_1\times d_2\times d_3}$, target rank $\r = (r_1,r_2,r_3)$}
			\STATE{$\U = \svd_{r_1}(\calM_1(\bcalT))$}
			\STATE{$\V = \svd_{r_2}(\calM_2(\bcalT))$}
			\STATE{$\W = \svd_{r_3}(\calM_3(\bcalT))$}
			\STATE{$\bcalC = \bcalT\times_1\U^\top\times_2\V^\top\times_3\W^\top$}
			\STATE{\textbf{Output: }$\bcalC\times_1\U\times_2\V\times_3\W$}
		\end{algorithmic}
	\end{algorithm}

	\section{Additional Real Data Applications}
	In this section, we examine the performance of \approach in two real world datasets: traffic data, and Walmart sales data. 
	We report the empirical behavior of \approach, TC, and COSTCO \citep{ibriga2022covariate} on two additional datasets. At the lowest sampling rate considered in the traffic experiment, TC and COSTCO do not converge in our implementation, whereas \approach does. At higher sampling rates, \approach is competitive with the compared methods. These comparisons are specific to the datasets and implementations considered here.
	
	\subsection{Traffic Data}
	PeMS  (Performance Measurement System) collected the traffic data from over 39,000 individual detectors located across all major metropolitan areas of the State of California \footnote{The dataset is publicly available in \url{https://pems.dot.ca.gov}}.
	We sub-sampled the speed data from 21 road segments over 207 days, where the data are sampled at 5-minute interval.
	And we obtain an original tensor data of size $288\times 21\times 207$. 
	Our aim is to complete the future data with the assistance of old data as side information. 
	More specifically, we split the data into a fully observed matrix $\M_s\in\RR^{288\times 21}$ containing the data on the first day, and the remaining tensor data of size $\bcalT^*\in\RR^{288\times 21\times 206}$ with $\alpha$ of its entries revealed, where $\alpha\in\{1\%, 3\%, 5\%\}$. We test the performance of different algorithms by comparing the relative errors. 
	In fact, the side information and the original tensor are coupled along both first and the second dimension of the tensor $\bcalT^*$, we test the performances of \approach by extracting the subspace along 1st dimension/2nd dimension/both dimensions. 
	We measure the performances of different algorithms using the relative error $\fro{\hat\bcalT - \bcalT^*} / \fro{\bcalT^*}$, where $\hat\bcalT$ is the output of the algorithm. 
	We also plot the performance of TC and COSTCO. 
	In our implementation, COSTCO and TC do not converge when only $1\%$ of the entries are revealed. We plot their relative errors versus iteration in Figure \ref{fig:traffic-TCCOSTCO}. \approach converges in this experiment, and Table \ref{table:traffic:relerr} reports its final relative error for three ways of using the side information.
	
	\begin{figure}[htbp]
		\begin{minipage}{0.32\textwidth}
			\centering
			\includegraphics[width=\linewidth]{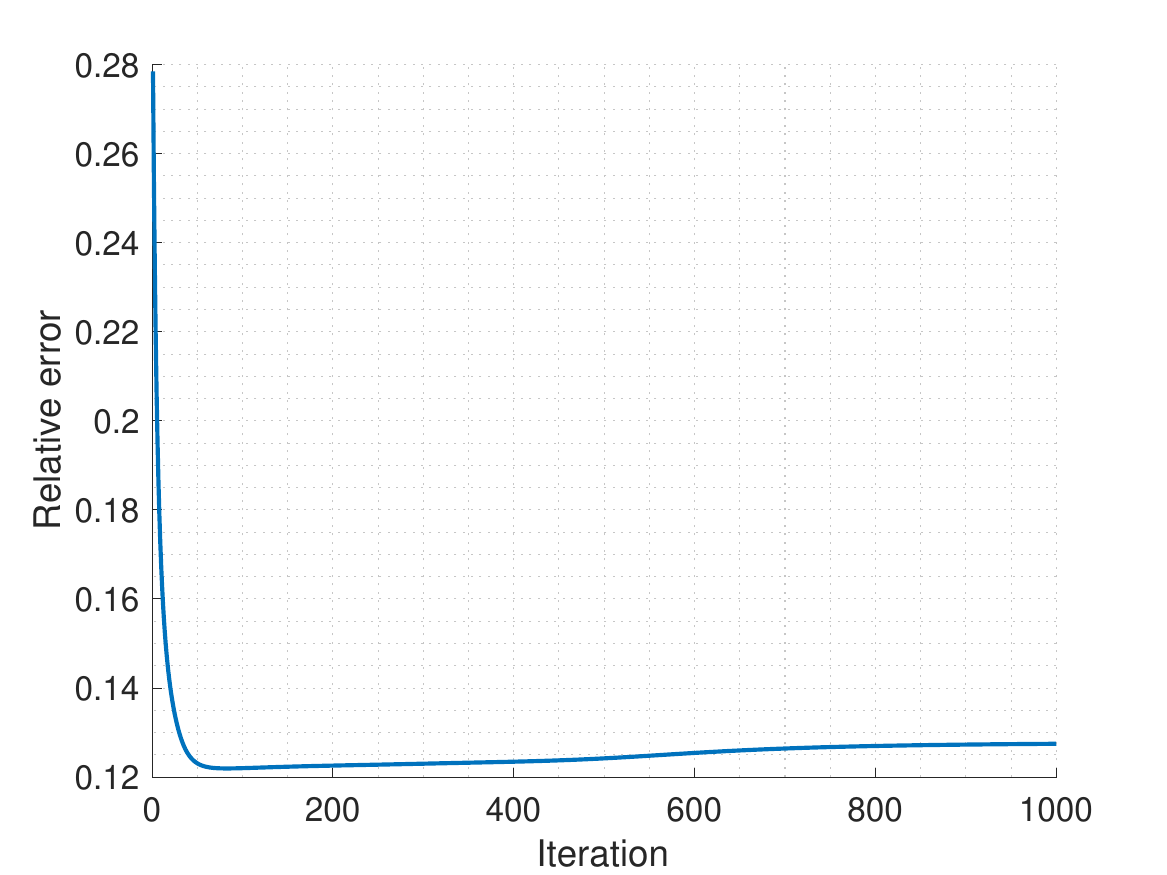}
		\end{minipage}
		\hfill
		\begin{minipage}{0.32\textwidth}
			\centering
			\includegraphics[width=\linewidth]{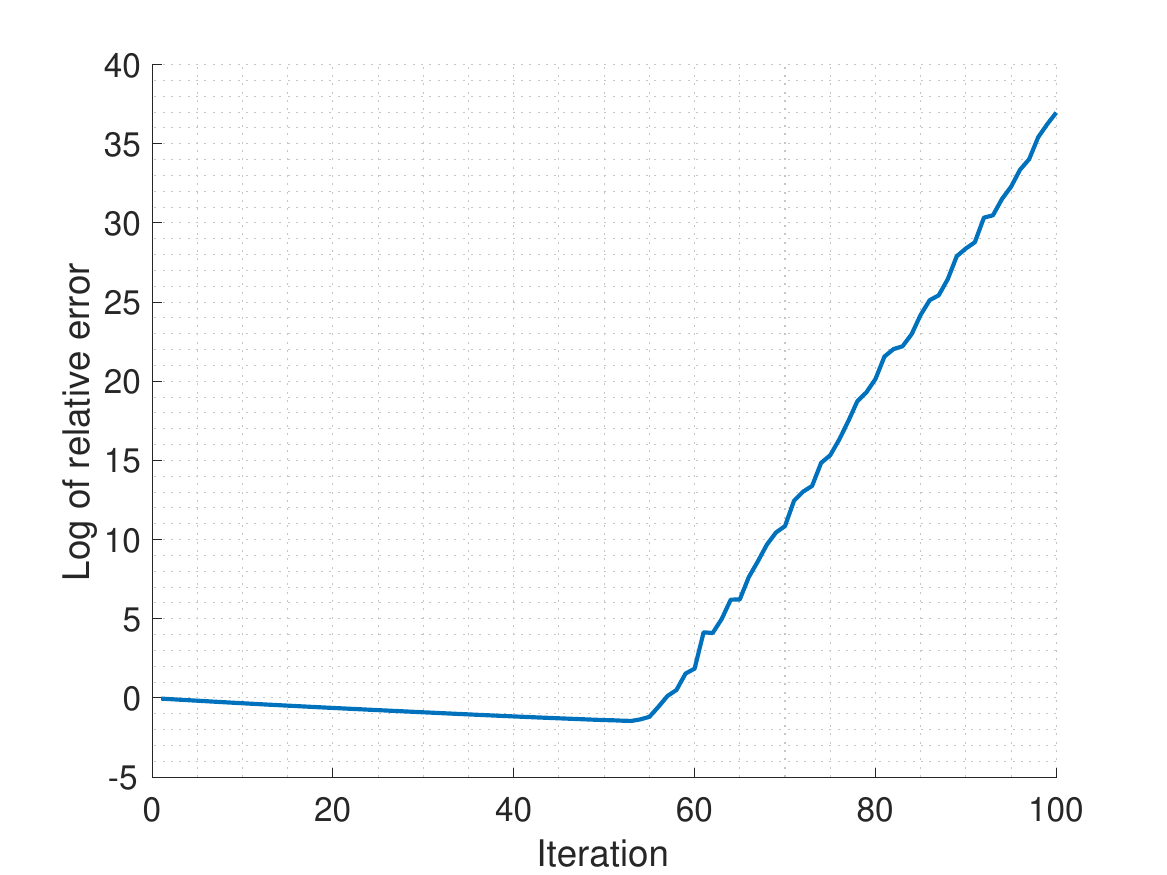}
		\end{minipage}
		\hfill
		\begin{minipage}{0.32\textwidth}
			\centering
			\includegraphics[width=\linewidth]{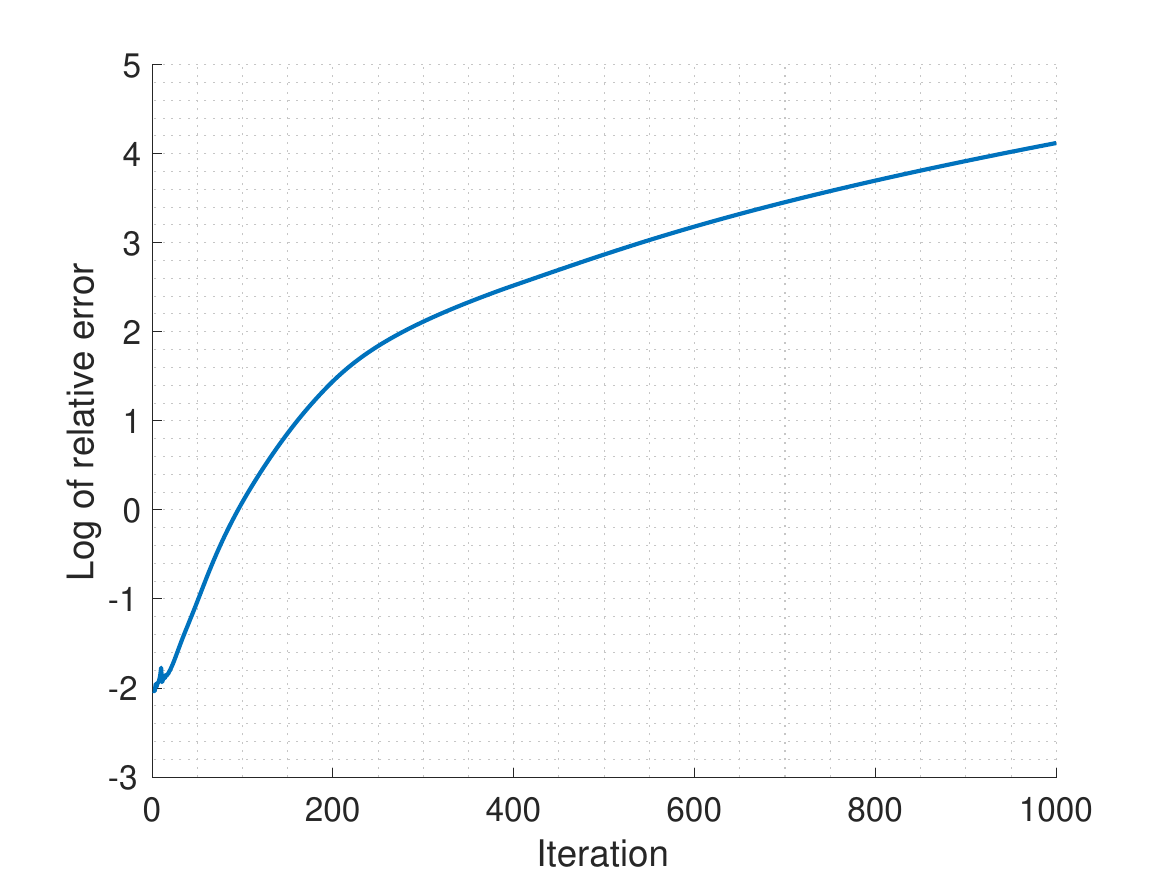}
		\end{minipage}
		\caption{Relative error using \approach(left), TC (middle), and COSTCO (right). Notice the y-axis for the latter two are in logarithm. }
		\label{fig:traffic-TCCOSTCO}
	\end{figure}
	\begin{table}[htbp]
		\centering
		\begin{tabular}{c||c|c|c|c|c}
			Algorithm & \approach 1st & \approach 2nd & \approach 1st \& 2nd&COSTCO & TC\\
			\hline 
			1\% & 0.341& 0.314  & \textbf{0.139}& NA&NA\\
			\hline
			3\% & 0.152&\textbf{0.094}&0.111&0.097&0.117  \\
			\hline
			5\% & 0.127&\textbf{0.088}&0.108& 0.09&0.115
		\end{tabular}
		\caption{Comparison of \approach with COSTCO and TC. For \approach, the subspaces are extracted along the 1st dimension/2nd dimension/1st \& 2nd dimension.}
		\label{table:traffic:relerr}
	\end{table}
	These results suggest that \approach can be useful in this dataset when observations are limited and informative side information is available. At the $3\%$ and $5\%$ sampling rates, its performance is comparable with COSTCO and better than TC in the reported experiments. At the $1\%$ sampling rate, using subspaces from both the first and second modes yields a smaller relative error ($0.139$) than using either subspace alone ($0.341$ and $0.314$). Although using more subspaces is not guaranteed to improve performance, the method can be adapted to incorporate subspace information from multiple modes.

	\subsection{Walmart Sales Prediction}
	\emph{Walmart Inc.} is an American multinational retail corporation that operates a chain of hypermarkets (also called supercenters), discount department stores, and grocery stores in the United States.
	In this dataset\footnote{The dataset is publicly available in \url{https://www.kaggle.com/datasets/varsharam/walmart-sales-dataset-of-45stores/data}}, we have weekly sales of 45 stores for the year 2010-2012 including the factors affecting sales such as holidays, and unemployment rate.
	After preprocessing the data, we obtain a tensor of size $45\times 143 \times 3$. The first dimension is the id of the store, the second dimension represents the week, and the last dimension is the unemployment level, which we convert the rate to a category variable \{Low, Medium, High\}. 
	The $(i,j,k)$-th entry of the tensor represents the sales of $i$-th store within $j$-th week under the unemployment level $k$. 
	We collect the data from one store as our side information, which is an incomplete matrix $\tilde\M_s\in\RR^{143\times 3}$, and the sales data from the rest stores as our target tensor $\bcalT^*\in\RR^{44\times 143\times 3}$.  
	In order to obtain the subspace from $\tilde \M_s$, we first use standard matrix completion algorithm to complete the matrix. We denote the completed matrix by $\M_s$, and the subspace is extracted using SVD, $\U_s$ is the top $r_1$ left singular vectors of $\M_s$. 
	Now the tensor $\bcalT^*$ is incomplete, we randomly sample $10\%$ of its entries as test set, which we denote by $\Omega_{\text{test}}$.
	After removing the test set, the tensor has roughly $30\%$ of its entries revealed. 
	The quality of the estimator $\hat\bcalT$ is measured by the averaged $\ell_2$ error: $\fro{\calP_{\Omega_{\text{test}}}(\bcalT^* - \hat\bcalT)}/\sqrt{|\Omega_{\text{test}}|}$, and the relative error $\fro{\calP_{\Omega_{\text{test}}}(\bcalT^* - \hat\bcalT)}/\fro{\calP_{\Omega_{\text{test}}}(\bcalT^*)}$. 
	We also compare our algorithm with COSTCO and TC. The results are summarized in Table \ref{table:compare:walmart}. 
	\begin{table}[htbp]
		\centering
		\begin{tabular}{c||c|c|c}
			Algorithm & \approach & COSTCO & TC\\
			\hline 
			averaged $\ell_2$ error & \textbf{0.246}& 0.312  &0.353\\
			\hline
			relative error & \textbf{0.350} &0.443 & 0.501
		\end{tabular}
		\caption{Comparison of \approach with COSTCO and TC. }
		\label{table:compare:walmart}
	\end{table}
	In this experiment, both methods using side information outperform TC, and \approach has lower held-out errors than COSTCO.

	\section{Derivation of Bounds for \texorpdfstring{$\gamma$}{gamma}}\label{app:gamma}
	Let $\U, \V \in \RR^{d \times r}$ be two column orthonormal matrices ($\U^\top \U = \V^\top \V = \I_r$). We consider two distance metrics (projective distance $\sf{d_p}$ and Procrustes distance $\sf{d_{proc}}$) between the subspaces spanned by $\U$ and $\V$:
	\begin{align*}
		\sf{d_{p}}:= \fro{\U\U^\top - \V\V^\top},\\
		\sf{d_{proc}}:=\min_{\R\in\OO_r}\fro{\U - \V\R}.
	\end{align*}
	Then simple linear algebra shows (see e.g. \cite{chen2021spectral})
	\begin{align*}
		\frac{1}{\sqrt{2}}\sf{d_p}\leq \sf{d_{proc}} \leq \sf{d_p}. 
	\end{align*}
	Now the quantity $\gamma = \frac{\sqrt{d}}{\sf{d_p}}\twoinf{\U - \V\R}$, where $\R = \arg\min_{\R\in\OO_r}\fro{\U - \V\R}$, can be equivalently written as 
	\begin{align*}
		\gamma = \frac{\sf{d_{proc}}}{\sf{d_p}}\cdot \frac{\sqrt{d}\twoinf{\U - \V\R}}{\fro{\U-\V\R}}. 
	\end{align*}
	From the definition of $\twoinf{\cdot}$, we see $\frac{\sqrt{d}\twoinf{\U - \V\R}}{\fro{\U-\V\R}}\in[1, \sqrt{d}]$. Together with the fact $\frac{\sf{d_{proc}}}{\sf{d_p}}\in[\frac{1}{\sqrt{2}}, 1]$, we conclude $\gamma\in[\frac{1}{\sqrt{2}},\sqrt{d}]$. 
	
	\noindent\textbf{Example}\par
	We illustrate two representative geometric configurations corresponding to row-concentrated and uniformly distributed subspace errors. Let $\{\e_1, \cdots, \e_d\}$ denote the standard basis vectors in $\RR^d$.
	
	\noindent\textit{The Spiky Case ($\gamma = \sqrt{d}$).}This upper bound is achieved when the subspace estimation error is concentrated in a single row. Let $\u = \frac{1}{\sqrt{2}}(\e_1+\e_2)$ and $\v = \frac{1}{\sqrt{2}}(\e_1-\e_2)$. 
	Then $\twoinf{\u-\v}= \sqrt{2}$, $\fro{\u\u^\top - \v\v^\top} = \sqrt{2}$, which implies $\gamma = \sqrt{d}$. 
	
	\noindent\textit{The Incoherent Case ($\gamma = 1$).}This value is achieved when the error energy is spread uniformly across all coordinates. Let $d$ be an even number. Consider two vectors with disjoint support:$$\u = \sqrt{\frac{2}{d}} \sum_{k=1}^{d/2} \e_{2k-1}, \quad \v = \sqrt{\frac{2}{d}} \sum_{k=1}^{d/2} \e_{2k}.$$
	Then $\twoinf{\u-\v} = \sqrt{\frac{2}{d}}$ and $\fro{\u\u^\top - \v\v^\top}^2 = 2$, equivalently $\fro{\u\u^\top - \v\v^\top} = \sqrt{2}$, which implies $\gamma =1$. 

	\section{Proofs}
	Lemma \ref{lemma:ptperp} - Lemma \ref{lemma:ydelta} and Theorem \ref{thm:concentration:psi2} used in this section can be found in Section \ref{sec:lemmas}. 
	
	\subsection{Proof of Theorem \ref{thm:init}}
	In this section, we compute the upper bound for $\fro{\tilde\M_0 - \M^*}$. 
	Notice $$d^*\EE y_i\Y_i = \M^*(\I_{d_3}\otimes \U^{*\top}\U_s) =: \N. $$
	We have
	\begin{align*}
		\fro{\tilde\M_0 - \M^*} \leq \fro{\tilde\M_0 - \N} + \fro{\N - \M^*}. 
	\end{align*}
	For the first term above, we have
	\begin{align*}
		\fro{\tilde\M_0 - \N} \leq \sqrt{2r}\op{\tilde\M_0 - \N} \leq 2\sqrt{2r} \op{\frac{d^*}{n}\sum_{i=1}^ny_i\Y_i - \N}.
	\end{align*}
	Next we consider the upper bound for $\op{\frac{d^*}{n}\sum_{i=1}^ny_i\Y_i - \N}$. Recall $y_i = \inp{\bcalE_{\omega_i}}{\bcalT^*} + \epsilon_i$, we have
	\begin{align*}
		\frac{d^*}{n}\sum_{i=1}^ny_i\Y_i - \N = \frac{d^*}{n}\sum_{i=1}^n \inp{\bcalE_{\omega_i}}{\bcalT^*}\Y_i - \N + \frac{d^*}{n}\sum_{i=1}^n\epsilon_i\Y_i.
	\end{align*}
	
	\noindent\textit{Upper bound for $\op{\frac{d^*}{n}\sum_{i=1}^n \inp{\bcalE_{\omega_i}}{\bcalT^*}\Y_i - \N}. $}
	Recall $\Y_i = \e_{k_i}(\e_{g_i}\otimes \e_{l_i})^\top(\I_{d_3}\otimes \U_s)$. And the incoherence implies $|\inp{\bcalE_{\omega_i}}{\bcalT^*}|\leq \sqrt{\frac{\mu^3 r^*}{d^*}}\lambda_{\max}$. 
	Therefore the upper bound of $\op{d^*\inp{\bcalE_{\omega_i}}{\bcalT^*}\Y_i - \N}$ is as follows:
	\begin{align*}
		\op{d^*\inp{\bcalE_{\omega_i}}{\bcalT^*}\Y_i - \N} \leq 2\sqrt{\mu^3\mu_s r^*r_1d_2d_3}\lambda_{\max}. 
	\end{align*}
	On the other hand, we have
	\begin{align*}
		\EE (d^*\inp{\Y_i}{\M^*})^2\Y_i \Y_i^\top &\leq \mu^3d^*r^*\lambda_{\max}^2 \frac{1}{d^*}\sum_{l,k,g}\e_k(\e_{g}\otimes\e_l)^\top(\I\otimes \U_s\U_s^{\top})(\e_{g}\otimes\e_l)\e_k^\top\\
		&= \mu^3r^*\lambda_{\max}^2 d_3r_1\I_{d_2},
	\end{align*}
	where $\A\leq \B$ means $\B-\A$ is PSD matrix. Similarly, we can show, 
	\begin{align*}
		\EE (d^*\inp{\Y_i}{\M^*})^2 \Y_i^\top\Y_i &\leq\mu^3r^*\lambda_{\max}^2 d_2\I_{d_3r_1},
	\end{align*}
	Using matrix Bernstein inequality, we obtain with probability exceeding $1-d^{-10}$, 
	\begin{align}\label{bound1}
		\op{\frac{d^*}{n}\sum_{i=1}^n \inp{\bcalE_{\omega_i}}{\bcalT^*}\Y_i - \N} \lesssim \bigg(\frac{\sqrt{\mu_s\mu^3 r_1d_2d_3r^*}}{n}\log d + \sqrt{\frac{\mu^3r^* (d_3r_1\vee d_2)}{n}\log d}\bigg)\lambda_{\max}.
	\end{align}

	Next we use Theorem \ref{thm:concentration:psi2} to bound $\op{\frac{d^*}{n}\sum_{i=1}^n\epsilon_i\Y_i}$. Notice
	\begin{align*}
		\psitwo{\op{\epsilon_i\Y_i}} \leq \sigma\cdot \sqrt{\frac{\mu_s r_1}{d_1}}.
	\end{align*}
	And 
	\begin{align*}
		\max\bigg\{\bigg\|\frac{1}{n}\sum_{i=1}^n\EE\epsilon_i^2\Y_i\Y_i^\top\bigg\|^{1/2}, \bigg\|\frac{1}{n}\sum_{i=1}^n\EE\epsilon_i^2\Y_i^\top\Y_i\bigg\|^{1/2}\bigg\}= \sigma^2\frac{1}{d^*}(d_2\vee d_3r_1). 
	\end{align*}
	From Theorem \ref{thm:concentration:psi2}, we see with probability exceeding $1-d^{-10}$, 
	\begin{align}\label{bound2}
		\op{\frac{d^*}{n}\sum_{i=1}^n\epsilon_i\Y_i} \lesssim \bigg(\sqrt{\frac{d^*}{n}(d_2\vee d_3 r_1)\log d} + \frac{\log^{3/2}d}{n}\sqrt{\mu_s r_1d_2d_3d^*}\bigg)\sigma. 
	\end{align}
	We conclude from \eqref{bound1}, \eqref{bound2},
	\begin{align*}
		\op{\frac{d^*}{n}\sum_{i=1}^ny_i\Y_i - \M^*} &\lesssim\bigg(\frac{\sqrt{\mu_s\mu^3 r_1d_2d_3r^*}}{n}\log d + \sqrt{\frac{\mu^3r^* (d_3r_1\vee d_2)}{n}\log d}\bigg)\lambda_{\max}\\
		&\quad + \bigg(\sqrt{\frac{d^*}{n}(d_2\vee d_3 r_1)\log d} + \frac{\log^{3/2}d}{n}\sqrt{\mu_s r_1d_2d_3d^*}\bigg)\sigma.
	\end{align*}
	Finally we consider $\fro{\M^* - \M^*(\I_{d_3}\otimes \U^{*\top}\U_s)}$:
	\begin{align*}
		\fro{\M^* - \M^*(\I_{d_3}\otimes \U^{*\top}\U_s)} &= \fro{\M^*(\I_{d_3}\otimes \big(\I_{r_1} - \U^{*\top}\U_s)\big)}\\
		&= \fro{\bcalC^*\times_1\big(\I_{r_1} - \U^{*\top}\U_s)\big) \times_2 \V^*\times_3\W^*}\\
		&\leq \lambda_{\max}\cdot \delta. 
	\end{align*}
	In conclusion, we have 
	\begin{align*}
		\fro{\tilde\M_0 - \M^*} &\leq C\sqrt{r}\bigg(\frac{\sqrt{\mu_s\mu^3 r_1d_2d_3r^*}}{n}\log d + \sqrt{\frac{\mu^3r^* (d_3r_1\vee d_2)}{n}\log d}\bigg)\lambda_{\max}\\
		&\quad + C\sqrt{r}\bigg(\sqrt{\frac{d^*}{n}(d_2\vee d_3 r_1)\log d} + \frac{\log^{3/2}d}{n}\sqrt{\mu_s r_1d_2d_3d^*}\bigg)\sigma  +\lambda_{\max}\cdot \delta\\
		&=:\tau.
	\end{align*}
	
	We denote $\tilde\M_0 = \tilde\L_0\tilde\S_0\tilde\R_0^\top, \M^* = \L^*\S^*\R^{*\top}$ be their compact rank $r$ SVD.  
	Now that we obtain $\fro{\tilde\M_0 - \M^*}\leq \tau$, 
	and Under the assumptions in Theorem \ref{thm:init}, $\tau\leq\frac{1}{2}\kappa^{-1}\lambda_{\min}$.
	So from Wedin's sin$\Theta$ Theorem, we have 
	\begin{align*}
		\fro{\tilde\L_0\tilde\L_0^\top - \L^*\L^{*\top}} \leq \frac{2\tau}{\lambda_{\min}}, \quad \fro{\tilde\R_0\tilde\R_0^\top - \R^*\R^{*\top}} \leq \frac{2\tau}{\lambda_{\min}}. 
	\end{align*}
	Now from Lemma \ref{lemma:truncation}, we obtain $\L_0,\R_0$ satisfying 
	\begin{align*}
		&\incoh(\L_0), \incoh(\R_0) \leq 3\mu,\\
		&\fro{\L_0\L_0^\top - \L^*\L^{*\top}} \leq \frac{8\pi\tau}{\lambda_{\min}}, \\
		& \fro{\R_0\R_0^\top - \R^*\R^{*\top}} \leq \frac{8\pi\tau}{\lambda_{\min}}.  
	\end{align*}
	Recall $\M_0 = \L_0\L_0^\top\tilde\M_0\R_0\R_0^\top$, and thus 
	\begin{align*}
		\fro{\M_0 - \M^*} &\leq \fro{\L_0\L_0^\top-\L^*\L^{*\top}}\op{\tilde\M_0} + \fro{\tilde\M_0 - \M^*} + \fro{\R_0\R_0^\top-\R^*\R^{*\top}}\op{\M^*} \\
		&\lesssim \kappa \tau. 
	\end{align*}
	Under the assumptions in Theorem \ref{thm:init}, $\fro{\M_0-\M^*}\lesssim \lambda_{\min}$. 

	\subsection{Proof of Theorem \ref{thm:convergence}}
	We have a good initial estimator $\M_0$ such that $\incoh(\M_0)\leq 3\mu$ and $\fro{\M_0 - \M^*}\leq c_0\lambda_{\min}$. 
	We use induction and first we list some properties $\M_l$ satisfies:
	\begin{enumerate}
		\item $\M_l$ is close to $\M^*$, i.e., $\fro{\M_l-\M^*}\leq c_0\lambda_{\min}$. This also implies $\lambda_{\max}(\M_l)\leq 1.01\lambda_{\max}$, $\lambda_{\min}(\M_l)\geq0.99\lambda_{\min}$, $\kappa(\M_l)\leq 1.03\kappa$. 
		\item $\incoh(\M_l)\leq 4\mu^2\kappa^2r=:\mu_1$
	\end{enumerate}
	
	We first consider the upper bound for $\fro{\M_{l} - \M^* - \eta \calP_{\TT_{l}}\G_l}^2$. Notice
	\begin{align*}
		\fro{\M_{l} - \M^* - \eta \calP_{\TT_{l}}\G_l}^2 = \fro{\M_l-\M^*}^2 -2\eta\inp{\M_{l} - \M^*}{\calP_{\TT_{l}}\G_l} + \eta^2\fro{\calP_{\TT_{l}}\G_l}^2.
	\end{align*}
	Recall $\G_l = \calY^*(\calY(\M_l)-\y) =  \calY^*\calY(\M_l-\M^*)-\calY^*(\bdelta + \beps)$, and thus 
	\begin{align*}
		&\quad\inp{\M_{l} - \M^*}{\calP_{\TT_{l}}\G_l} = \inp{\M_{l} - \M^*}{\G_l} - \inp{\calP_{\TT_{l}}^{\perp}(\M_{l} - \M^*)}{\G_l}\\
		&=\underbrace{\ltwo{\calY(\M_l-\M^*)}^2}_{\textsf{B}_1}-\underbrace{\inp{\M_{l} - \M^*}{\calY^*(\bdelta + \beps)}}_{\textsf{B}_2} -
		\underbrace{\inp{\calP_{\TT_{l}}^{\perp}(\M_{l} - \M^*)}{\calY^*\calY(\M_l-\M^*)}}_{\textsf{B}_3} \\
		&\quad+ \underbrace{\inp{\calP_{\TT_{l}}^{\perp}(\M_{l} - \M^*)}{\calY^*(\bdelta + \beps)}}_{\textsf{B}_4}. 
	\end{align*}
	\noindent$(\textsf{B}_1)$ For the first term, we have 
	\begin{align}\label{eq:ym-m}
		\ltwo{\calY(\M_l-\M^*)}^2 \geq 0.99\ltwo{\calY\calP_{\TT^*}(\M_l-\M^*)}^2 - 99\ltwo{\calY\calP_{\TT^*}^{\perp}(\M_l-\M^*)}^2,
	\end{align}
	where we use the inequality $(a+b)^2 \geq 0.99a^2 - 99b^2$. 
	From Lemma \ref{lemma:pty}, when $n\gtrsim \mu\mu_s(r_1d_2r_3 + r_1r_2d_3)\log d$, with probability exceeding $1-d^{-10}$, 
	\begin{align}\label{eq:yptm-m}
		\ltwo{\calY\calP_{\TT^*}(\M_l-\M^*)}^2 \geq 0.99\frac{n}{d^*}\fro{\calP_{\TT^*}(\M_l-\M^*)}^2. 
	\end{align}
	We consider the following empirical process:
	\begin{align*}
		\beta_n(\gamma_1,\gamma_2) := \sup_{\M\in\KK_{\gamma_1,\gamma_2}}\bigg|\ltwo{\calY(\M)}^2 - \frac{n}{d^*}\fro{\M}^2\bigg|,
	\end{align*}
	where 
	\begin{align*}
		\KK_{\gamma_1,\gamma_2} = \big\{\M\in\RR^{d_2\times d_3r_1}: \fro{\M} = 1, \linf{\M}\leq \gamma_1, \nuc{\M}\leq \gamma_2\big\}.
	\end{align*}
	On the other hand, 
	\begin{align}\label{eq:yptperp}
		\ltwo{\calY\calP_{\TT^*}^{\perp}(\M_l-\M^*)}^2 \leq \frac{n}{d^*}\fro{\calP_{\TT^*}^{\perp}(\M_l)}^2 + \fro{\calP_{\TT^*}^{\perp}(\M_l)}^2\cdot\beta_n\bigg(\frac{\linf{\calP_{\TT^*}^{\perp}\M_l}}{\fro{\calP_{\TT^*}^{\perp}\M_l}},\frac{\nuc{\calP_{\TT^*}^{\perp}\M_l}}{\fro{\calP_{\TT^*}^{\perp}\M_l}}\bigg)
	\end{align}
	We now consider the upper bound for $\linf{\calP_{\TT^*}^{\perp}\M_l}$, $\nuc{\calP_{\TT^*}^{\perp}\M_l}$. For the infinity norm bound, we have
	\begin{align*}
		\linf{\calP_{\TT^*}^{\perp}\M_l} \leq \linf{\M_l} + \linf{\calP_{\TT^*}\M_l}. 
	\end{align*}
	Since $\incoh(\M_l)\leq \mu_1$, we have 
	\begin{align*}
		\linf{\M_l} \leq \sqrt{\frac{\mu_1 r}{d_2}}\sqrt{\frac{\mu_1 r}{d_3r_1}}\op{\M_l}\leq 1.01\sqrt{\frac{\mu_1 r}{d_2}}\sqrt{\frac{\mu_1 r}{d_3r_1}}\lambda_{\max}. 
	\end{align*}
	On the other hand, 
	\begin{align*}
		\calP_{\TT^*}\M_l = \V^*\V^{*\top}\M_l  +\M_l (\W^*\W^{*\top}\otimes \I) - \V^*\V^{*\top}\M_l(\W^*\W^{*\top}\otimes \I).
	\end{align*}
	And 
	\begin{align*}
		\linf{\V^*\V^{*\top}\M_l} &\leq \sqrt{\frac{\mu r_2}{d_2}}\cdot \sqrt{\frac{\mu_1 r}{d_3r_1}}\op{\M_l}, \\
		\linf{\M_l (\W^*\W^{*\top}\otimes \I)} &\leq \sqrt{\frac{\mu r_3}{d_3}}\cdot \sqrt{\frac{\mu_1 r}{d_2}}\op{\M_l},\\
		\linf{\V^*\V^{*\top}\M_l(\W^*\W^{*\top}\otimes \I)} &\leq \sqrt{\frac{\mu r_2}{d_2}}\sqrt{\frac{\mu r_3}{d_3}}\cdot\op{\M_l}.
	\end{align*}
	As a result,
	\begin{align*}
		\linf{\calP_{\TT^*}^{\perp}\M_l} \leq 4.04\sqrt{\frac{\mu_1^2 r_2r_3}{d_2d_3}}\lambda_{\max}.
	\end{align*}
	On the other hand, since $\rank(\calP_{\TT^*}^{\perp}\M_l) = r$, we have together with Lemma \ref{lemma:ptperp},
	\begin{align*}
		\nuc{\calP_{\TT^*}^{\perp}\M_l}\leq \sqrt{r}\fro{\calP_{\TT^*}^{\perp}\M_l}\lesssim \sqrt{r}\frac{\fro{\M_l-\M^*}^2}{\lambda_{\min}}.
	\end{align*}
	And from Lemma \ref{lemma:empirical}, we see with probability exceeding $1-d^{-10}$, 
	\begin{align*}
		&\quad\fro{\calP_{\TT^*}^{\perp}\M_l}^2\cdot\beta_n\bigg(\frac{\linf{\calP_{\TT^*}^{\perp}\M_l}}{\fro{\calP_{\TT^*}^{\perp}\M_l}},\frac{\nuc{\calP_{\TT^*}^{\perp}\M_l}}{\fro{\calP_{\TT^*}^{\perp}\M_l}}\bigg) \\
		&\lesssim \frac{\sqrt{\mu_s}r_1}{\sqrt{d_1}}\linf{\calP_{\TT^*}^{\perp}\M_l}\nuc{\calP_{\TT^*}^{\perp}\M_l}\bigg(\sqrt{\frac{n}{d^*}(d_3r_1\vee d_2) \log d} + \sqrt{\frac{\mu_s r_1}{d_1}}\log d\bigg) \\
		&\quad + \sqrt{n\frac{\mu_s r_1^2}{d_1d^*}\log d}\cdot\linf{\calP_{\TT^*}^{\perp}\M_l}\cdot\fro{\calP_{\TT^*}^{\perp}\M_l} + \frac{\mu_s r_1^2}{d_1}\linf{\calP_{\TT^*}^{\perp}\M_l}^2\log d\\
		&\lesssim \sqrt{\frac{n\mu_s r_1^2}{d_1d^*}(d_3r_1\vee d_2) \log d} \cdot\linf{\calP_{\TT^*}^{\perp}\M_l}\cdot\fro{\calP_{\TT^*}^{\perp}\M_l}\\
		&\lesssim \frac{\sqrt{n}}{d^*}\sqrt{\mu_s\mu_1^2\kappa^2}\sqrt{(d_3r_1\vee d_2)r_1r^*}\sqrt{\log d}\cdot \fro{\M_l-\M^*}^2\\
		&\lesssim \frac{n}{d^*}\fro{\M_l-\M^*}^2.
	\end{align*} 
	where the second inequality holds as long as $n\gtrsim \mu_s (r\vee r_1)(d_2\wedge d_3r_1)\log d$, and the last inequality holds as long as $n\gtrsim (d_3r_1\vee d_2)r_1r^*\cdot \mu_s\mu_1^2\kappa^2\log d$.
	Together with \eqref{eq:yptperp}, we see
	\begin{align*}
		\ltwo{\calY\calP_{\TT^*}^{\perp}(\M_l-\M^*)}^2\leq 0.0001\frac{n}{d^*}\fro{\M_l-\M^*}^2.
	\end{align*}
	Together with \eqref{eq:ym-m} and \eqref{eq:yptm-m}, we see 
	\begin{align*}
		\ltwo{\calY(\M_l-\M^*)}^2 \geq 0.98\frac{n}{d^*}\fro{\M_l-\M^*}^2. 
	\end{align*}

	\noindent$(\textsf{B}_3)$ For the third term, we shall use Lemma \ref{lemma:YYPt} to bound it. Recall the notation $\M_l = \L_l\S_l\R_l^\top$, $\M^* = \L^*\S^*\R^{*\top}$, therefore
	$$\calP_{\TT_{l}}^{\perp}(\M_{l} - \M^*)= \L_{l,\perp}\L_{l,\perp}^{\top}\L^*{\S^*}^{1/2}{\S^*}^{1/2}\R^{*\top}\R_{l,\perp}\R_{l,\perp}^{\top}=:\B.$$
	And we decompose 
	\begin{align*}
		\M_l-\M^* &= \underbrace{(\L_l\L_l^\top-\L^*\L^{*\top})\M_l\R_l\R_l^\top}_{:=\A_1} + \underbrace{\L^*\L^{*\top}(\M_l-\M^*)\R_l\R_l^\top}_{:=\A_2} \\
        &\quad + \underbrace{\L^*\L^{*\top}\M^*(\R_l\R_l^\top-\R^*\R^{*\top})}_{:=\A_3}.
	\end{align*}
	Thus
	\begin{align*}
		\inp{\calP_{\TT_{l}}^{\perp}(\M_{l} - \M^*)}{\calY^*\calY(\M_l-\M^*)} &= \inp{\calP_{\TT_{l}}^{\perp}(\M_{l} - \M^*)}{\calY^*\calY(\M_l-\M^*)}\\
		&= \inp{\calY^*\calY\A_1}{\B} +  \inp{\calY^*\calY\A_2}{\B} +  \inp{\calY^*\calY\A_3}{\B}.
	\end{align*}
	From Lemma \ref{lemma:YYPt}, we see 
	\begin{align*}
		&\quad\bigg|\inp{\calP_{\TT_{l}}^{\perp}(\M_{l} - \M^*)}{\calY^*\calY(\M_l-\M^*)} - \frac{n}{d^*}\inp{\calP_{\TT_{l}}^{\perp}(\M_{l} - \M^*)}{\M_l-\M^*}\bigg|\\
		&\leq \bigg(\sqrt{\frac{n\mu_s r_1}{d^*d_1}(d_2\vee d_3r_1)\log d} + \frac{\mu_s r_1}{d_1}\log d\bigg)\cdot\big(\nuc{\calL(\A_1,\B)}+\nuc{\calL(\A_2,\B)}+\nuc{\calL(\A_3,\B)}\big)
	\end{align*}
	And from Lemma \ref{lemma:LAB}, we see 
	\begin{align*}
		\nuc{\calL(\A_1,\B)} &\lesssim r_1^{3/2}r\cdot\sqrt{\frac{\mu_1 r}{d_2}}\sqrt{\frac{\mu_1 r}{d_3r_1}}\lambda_{\max}\cdot\lambda_{\max}\frac{\fro{\M_l-\M^*}^2}{\lambda_{\min}^2}\\
		&\lesssim r_1r^2\mu_1\kappa^2\frac{\fro{\M_l-\M^*}^2}{\sqrt{d_2d_3}}. 
	\end{align*}
	And we can similarly bound $\nuc{\calL(\A_2,\B)}, \nuc{\calL(\A_3,\B)}$. So we conclude as long as $n \gtrsim (d_2\vee d_3r_1)r_1^3r^4 \mu_s\mu_1^2\kappa^4\log d$,
	we have 
	\begin{align*}
		\bigg|\inp{\calP_{\TT_{l}}^{\perp}(\M_{l} - \M^*)}{\calY^*\calY(\M_l-\M^*)} - \frac{n}{d^*}\inp{\calP_{\TT_{l}}^{\perp}(\M_{l} - \M^*)}{\M_l-\M^*}\bigg| \lesssim \frac{n}{d^*}\fro{\M_l-\M^*}^2.
	\end{align*}
	Together with Lemma \ref{lemma:ptperp}, we have 
	\begin{align}
		\bigg|\inp{\calP_{\TT_{l}}^{\perp}(\M_{l} - \M^*)}{\calY^*\calY(\M_l-\M^*)} \bigg| \leq 0.01\frac{n}{d^*}\fro{\M_l-\M^*}^2.
	\end{align}
	
	\noindent$(\textsf{B}_2,\textsf{B}_4)$ Notice $-\textsf{B}_2+\textsf{B}_4 = -\inp{\calP_{\TT_l}(\M_l-\M^*)}{\calY^*(\boldsymbol{\delta}+\boldsymbol{\epsilon})}$. Notice $\op{\calY^*(\boldsymbol{\delta})}$, $\op{\calY^*(\boldsymbol{\epsilon})}$ are bounded in Lemma \ref{lemma:ydelta}, \eqref{bound2} respectively, which give
	\begin{align}
		\op{\calY^*(\boldsymbol{\delta})} &\leq
		C\sqrt{\frac{\mu_s\mu^2 r^*}{d^*}}\frac{\gamma\delta \wedge \sqrt{\mu_s r_1}}{\sqrt{d_1}}\lambda_{\max} \log d + 1.01\frac{n}{d^*}\lambda_{\max}\delta
		\label{ydelta}\\
		\op{\calY^*(\boldsymbol{\epsilon})} &\lesssim \bigg(\sqrt{\frac{n}{d^*}(d_2\vee d_3 r_1)\log d} +\sqrt{\frac{\mu_s r_1}{d_1}} \log^{3/2}d\bigg)\sigma.\label{yeps}
	\end{align}
	Using Cauchy-Schwarz inequality, we see 
	\begin{align*}
		|-\textsf{B}_2+\textsf{B}_4| &\leq
		\nuc{\calP_{\TT_l}(\M_l-\M^*)}\cdot(\op{\calY^*(\boldsymbol{\delta})} + \op{\calY^*(\boldsymbol{\epsilon})})\\
		&\leq c_2\frac{n}{d^*}\fro{\M_l-\M^*}^2  + C_1r\frac{d^*}{n}(\op{\calY^*(\boldsymbol{\delta})}^2 + \op{\calY^*(\boldsymbol{\epsilon})}^2).
	\end{align*}
	Using \eqref{ydelta}, \eqref{yeps}, we have 
	\begin{align}\label{yd+e}
		\op{\calY^*(\boldsymbol{\delta})}^2 + \op{\calY^*(\boldsymbol{\epsilon})}^2&\lesssim
		\frac{\mu_s\mu^2 r^*}{d^*}\frac{\gamma^2\delta^2 \wedge \mu_s r_1}{d_1}\lambda_{\max}^2 \log^2 d + (\frac{n}{d^*})^2\lambda_{\max}^2\delta^2\notag\\
		&\quad + \bigg(\frac{n}{d^*}(d_2\vee d_3 r_1)\log d +\frac{\mu_s r_1}{d_1} \log^{3}d\bigg)\sigma^2.
	\end{align}
	And this leads to 
	\begin{align*}
		|-\textsf{B}_2+\textsf{B}_4| &\leq 
		0.01\frac{n}{d^*}\fro{\M_l-\M^*}^2  + C_1r\bigg(\frac{\mu_s\mu^2 r^*}{n}\frac{\gamma^2\delta^2 \wedge \mu_s r_1}{d_1}\lambda_{\max}^2 \log^2 d + \frac{n}{d^*}\lambda_{\max}^2\delta^2\\
		&\quad + \big((d_2\vee d_3 r_1)\log d +\frac{\mu_s d_2d_3r_1}{n} \log^{3}d\big)\sigma^2\bigg).
	\end{align*}

	In summary, we have 
	\begin{align}\label{ineq:inp}
		&\inp{\M_{l} - \M^*}{\calP_{\TT_{l}}\G_l} \geq 0.96\frac{n}{d^*}\fro{\M_l-\M^*}^2 - \sfR_1,
	\end{align}
	where 
	\begin{align*}
		\sfR_1 &= C_1r\bigg(\frac{\mu_s\mu^2 r^*}{n}\frac{\gamma^2\delta^2 \wedge \mu_s r_1}{d_1}\lambda_{\max}^2 \log^2 d + \frac{n}{d^*}\lambda_{\max}^2\delta^2
		+ \big((d_2\vee d_3 r_1)\log d +\frac{\mu_s d_2d_3r_1}{n} \log^{3}d\big)\sigma^2\bigg).
	\end{align*}
	
	\hspace{1cm}
	
	Next we consider the bound for $\fro{\calP_{\TT_{l}}\G_l}^2$. Recall $\G_l =  \calY^*\calY(\M_l-\M^*)-\calY^*(\bdelta + \beps)$. We have
	\begin{align*}
		\fro{\calP_{\TT_{l}}\G_l}^2 \leq (1+c_3)\fro{\calP_{\TT_{l}}\big(\calY^*\calY(\M_l-\M^*)\big)}^2 + (1+c_3^{-1})\fro{\calP_{\TT_{l}}\big(\calY^*(\bdelta + \beps)\big)}^2. 
	\end{align*}
	And 
	\begin{align*}
		\fro{\calP_{\TT_{l}}\big(\calY^*\calY(\M_l-\M^*)\big)}^2 &\leq (1+c_4^{-1})\fro{\calP_{\TT_{l}}\big((\calY^*\calY - \frac{n}{d^*}\calI)(\M_l-\M^*)\big)}^2\\
		&\quad + (1+c_4)\left(\frac{n}{d^*}\right)^2\fro{\calP_{\TT_{l}}(\M_l-\M^*)}^2.  
	\end{align*}
	Notice there exists some $\X_0$ with unit Frobenius norm, such that 
	\begin{align*}
		&\quad \fro{\calP_{\TT_{l}}\big((\calY^*\calY - \frac{n}{d^*}\calI)(\M_l-\M^*)\big)} = \inp{\calP_{\TT_{l}}\X_0}{(\calY^*\calY - \frac{n}{d^*}\calI)(\M_l-\M^*)}\notag\\
		&= \underbrace{\inp{\L_l\L_l^\top\X_0(\I- \R_l\R_l^\top)}{(\calY^*\calY - \frac{n}{d^*}\calI)(\M_l-\M^*)}}_{:=\sfC_1} + \underbrace{\inp{\X_0 \R_l\R_l^\top}{(\calY^*\calY - \frac{n}{d^*}\calI)(\M_l-\M^*)}}_{:=\sfC_2}. 
	\end{align*}
	For $\sfC_1$, we can decompose $\M_l - \M^* = \L^*\L^{*\top}(\M_l-\M^*)+ (\L_l \L_l^\top- \L^*\L^{*\top})\M_l$.
	Then we have (for notation simplicity, we assume $\L_l,\L^*$ are aligned)
	\begin{align*}
		\sfC_1 &= \underbrace{\inp{(\L_l-\L^*)\L_l^\top\X_0(\I- \R_l\R_l^\top)}{(\calY^*\calY - \frac{n}{d^*}\calI)( \L^*\L^{*\top}(\M_l-\M^*))}}_{:=\sfC_{1,1}} \\
		&\quad + \underbrace{\inp{\L^*\L_l^\top\X_0(\I- \R_l\R_l^\top)}{(\calY^*\calY - \frac{n}{d^*}\calI)( \L^*\L^{*\top}(\M_l-\M^*))}}_{:=\sfC_{1,2}}\\
		&\quad+ \underbrace{\inp{\L_l\L_l^\top\X_0(\I- \R_l\R_l^\top)}{(\calY^*\calY - \frac{n}{d^*}\calI)( (\L_l \L_l^\top- \L^*\L^{*\top})\M_l)}}_{:=\sfC_{1,3}}.
	\end{align*}
	We can use Lemma \ref{lemma:YYPt} and \ref{lemma:LAB} to control $\sfC_{1,1}$ and $\sfC_{1,3}$:
	\begin{align*}
		|\sfC_{1,1}|&\lesssim \bigg(\sqrt{\frac{n\mu_s r_1}{d^*d_1}(d_2\vee d_3r_1)\log d} + \frac{\mu_s r_1}{d_1}\log d\bigg) r_1^{3/2}r \cdot\sqrt{\frac{\mu r}{d_2}\frac{\mu_1 r}{d_3 r_1}}\kappa\cdot\fro{\M_l - \M^*}, \\
		|\sfC_{1,3}|&\lesssim \bigg(\sqrt{\frac{n\mu_s r_1}{d^*d_1}(d_2\vee d_3r_1)\log d} + \frac{\mu_s r_1}{d_1}\log d\bigg) r_1^{3/2}r \cdot\sqrt{\frac{\mu_1 r}{d_2}\frac{\mu_1 r}{d_3 r_1}}\kappa\cdot\fro{\M_l - \M^*}.
	\end{align*}
	And we use Lemma \ref{lemma:YYLR} for $\sfC_{1,2}$:
	\begin{align*}
		|\sfC_{1,2}|\lesssim \bigg(\frac{\mu_s r_1}{d_1}\frac{\mu r}{d_2}\log d + \sqrt{\frac{n}{d^*}\frac{\mu_s r_1}{d_1}\frac{\mu r}{d_2}\log d}\bigg)\fro{\M_l-\M^*}. 
	\end{align*}
	Putting these together, we obtain  
	\begin{align*}
		|\sfC_{1}| \lesssim \bigg(\sqrt{\frac{n\mu_s r_1}{d^*d_1}(d_2\vee d_3r_1)\log d} + \frac{\mu_s r_1}{d_1}\log d\bigg) r_1^{3/2}r \cdot\sqrt{\frac{\mu_1 r}{d_2}\frac{\mu_1 r}{d_3 r_1}}\kappa\cdot\fro{\M_l - \M^*}.
	\end{align*}
	For $\sfC_2$, we decompose $\M_l - \M^* = (\M_l-\M^*)\R^*\R^{*\top} + \M_l(\R_l\R_l^\top - \R^*\R^{*\top})$. And (assume $\R_l,\R^*$ are aligned)
	\begin{align*}
		\sfC_2 &= \inp{\X_0 \R_l(\R_l-\R^*)^\top}{(\calY^*\calY - \frac{n}{d^*}\calI)((\M_l-\M^*)\R^*\R^{*\top})}\\
		&\quad +  \inp{\X_0 \R_l\R^{*\top}}{(\calY^*\calY - \frac{n}{d^*}\calI)((\M_l-\M^*)\R^*\R^{*\top})}\\
		&\quad + \inp{\X_0 \R_l\R_l^\top}{(\calY^*\calY - \frac{n}{d^*}\calI)(\M_l(\R_l\R_l^\top - \R^*\R^{*\top}))}.
	\end{align*}
	We can similarly show that 
	\begin{align*}
		|\sfC_2| \lesssim \bigg(\sqrt{\frac{n\mu_s r_1}{d^*d_1}(d_2\vee d_3r_1)\log d} + \frac{\mu_s r_1}{d_1}\log d\bigg) r_1^{3/2}r \cdot\sqrt{\frac{\mu_1 r}{d_2}\frac{\mu_1 r}{d_3 r_1}}\kappa\cdot\fro{\M_l - \M^*}.
	\end{align*}
	
	In conclusion, 
	\begin{align*}
		&\quad \fro{\calP_{\TT_{l}}\big((\calY^*\calY - \frac{n}{d^*}\calI)(\M_l-\M^*)\big)} \\
		& \lesssim \bigg(\sqrt{\frac{n\mu_s r_1}{d^*d_1}(d_2\vee d_3r_1)\log d} + \frac{\mu_s r_1}{d_1}\log d\bigg) r_1^{3/2}r \cdot\sqrt{\frac{\mu_1 r}{d_2}\frac{\mu_1 r}{d_3 r_1}}\kappa\cdot\fro{\M_l - \M^*}.
	\end{align*}
	As long as $n\gtrsim (d_2\vee d_3r_1)\log d\cdot r_1^2r^4\mu_s\mu_1^2\kappa^2$, 
	\begin{align*}
		\fro{\calP_{\TT_{l}}\big((\calY^*\calY - \frac{n}{d^*}\calI)(\M_l-\M^*)\big)}\lesssim \frac{n}{d^*}\fro{\M_l - \M^*}, 
	\end{align*}
	which leads to 
	\begin{align*}
		\fro{\calP_{\TT_{l}}\big((\calY^*\calY)(\M_l-\M^*)\big)} \leq 1.001 \frac{n}{d^*}\fro{\M_l - \M^*}.
	\end{align*}
	
	Next we consider $\fro{\calP_{\TT_{l}}\big(\calY^*(\bdelta + \beps)\big)}^2$ using \eqref{ydelta} and \eqref{yeps}:
	\begin{align*}
		\fro{\calP_{\TT_{l}}\big(\calY^*(\bdelta)+\calY^*(\beps)\big)}^2 &\leq r\op{\calY^*(\bdelta)+ \calY^*(\beps)}^2\\
		&\lesssim r\bigg(\frac{\mu_s\mu^2 r^*}{d^*}\frac{\gamma^2\delta^2 \wedge \mu_s r_1}{d_1}\lambda_{\max}^2 \log^2 d + (\frac{n}{d^*})^2\lambda_{\max}^2\delta^2\bigg)\\
		&\quad +r \bigg(\frac{n}{d^*}(d_2\vee d_3 r_1)\log d+\frac{\mu_s r_1}{d_1} \log^{3}d\bigg)\sigma^2.
	\end{align*}
	In summary 
	\begin{align}\label{ineq:ptg}
		\fro{\calP_{\TT_{l}}(\G_l)}^2 \leq 1.01(\frac{n}{d^*})^2\fro{\M_l-\M^*}^2 + \sfR_2, 
	\end{align}
	where 
	\begin{align*}
		\sfR_2 &= r\bigg(\frac{\mu_s\mu^2 r^*}{d^*}\frac{\gamma^2\delta^2 \wedge \mu_s r_1}{d_1}\lambda_{\max}^2 \log^2 d + (\frac{n}{d^*})^2\lambda_{\max}^2\delta^2\bigg)
		+r \bigg(\frac{n}{d^*}(d_2\vee d_3 r_1)\log d+\frac{\mu_s r_1}{d_1} \log^{3}d\bigg)\sigma^2.
	\end{align*}
	Now we can summarize form \eqref{ineq:inp}, and \eqref{ineq:ptg}, 
	\begin{align*}
		\fro{\M_l-\M^*-\eta \calP_{\TT_{l}}(\G_l)}^2 \leq (1-1.92\eta  + 1.01\eta^2) \fro{\M_l-\M^*}^2 + 2\eta\sfR_1 + \eta^2 \sfR_2. 
	\end{align*}
	By setting $\eta = \frac{d^*}{n}$, the above becomes 
	\begin{align*}
		&\quad\fro{\M_l-\M^*-\eta \calP_{\TT_{l}}(\G_l)}^2 \leq 0.09\fro{\M_l-\M^*}^2 \\
		&+ C\bigg(\frac{d_2d_3 rr^*\mu_s\mu^2}{n^2}(\gamma^2\delta^2 \wedge \mu_s r_1)\lambda_{\max}^2\log^2 d + r\lambda_{\max}^2\delta^2 + \frac{d^*r(d_2\vee d_3r_1)}{n}\log d\cdot\sigma^2\bigg). 
	\end{align*}
	Now as long as 
	\begin{align*}
		r\lambda_{\max}^2\delta^2 &\lesssim \lambda_{\min}^2,\\
		d_2d_3 rr^*\mu_s\mu^2(\gamma^2\delta^2 \wedge \mu_s r_1)\kappa^2\log^2 d&\lesssim n^2, \\
		\frac{d^*r(d_2\vee d_3r_1)}{n}\log d&\lesssim (\lambda_{\min}/\sigma)^2, 
	\end{align*}
	we have $\fro{\M_l-\M^*-\eta \calP_{\TT_{l}}(\G_l)}^2\lesssim \lambda_{\min}^2$. And from Lemma \ref{lemma:perturbation}, we have 
	\begin{align*}
		\fro{\M_{l+1}-\M^*}^2 = \fro{\svd_r(\W_l)-\M^*}^2 \leq \fro{\W_l-\M^*}^2 + C_1\frac{\fro{\W_l-\M^*}^3}{\lambda_{\min}}.
	\end{align*}
	With the choice of $\zeta$, we see $\M^*\in\{\M:\linf{\M^*}\leq \zeta\}$, and $\trim_{\zeta}$ is the projection operator. We have $\fro{\W_l-\M^*}\leq \fro{\M_l-\eta\calP_{\TT_{l}}\G_l - \M^*}$. Therefore, 
	\begin{align*}
		\fro{\M_{l+1}-\M^*}^2
		&\leq 1.1 \fro{\M_{l} - \M^* - \eta \calP_{\TT_{l}}\G_l}^2\\
		&\leq 0.1\fro{\M_l-\M^*}^2 + C\bigg(\frac{d_2d_3 rr^*\mu_s\mu^2}{n^2}(\gamma^2\delta^2 \wedge \mu_s r_1)\lambda_{\max}^2\log^2 d \\
		&\hspace{6cm}+ r\lambda_{\max}^2\delta^2 + \frac{d^*r(d_2\vee d_3r_1)}{n}\log d\cdot\sigma^2\bigg).
	\end{align*}
	Finally we show the incoherence of $\M_{l+1}$ is bounded. Since $\zeta \geq \linf{\M^*}$, and $\fro{\M_l-\M^*-\eta \calP_{\TT_{l}}(\G_l)} \leq \frac{1}{2}\lambda_{\min}$, the conditions of Lemma \ref{lemma:spiki-incoh} are satisfied, and we conclude $\incoh(\M_{l+1})\leq 4\mu^2\kappa^2r\leq \mu_1$. 

	After $l_{\max} = \big\lceil \log\big(\frac{Cr\lambda_{\max}}{\fro{\M_0-\M^*}}\big)\big\rceil$ iterations, 
	we have 
	\begin{align*}
		\fro{\M_{l_{\max}} - \M^*}^2 &\lesssim \frac{d_2d_3 rr^*\mu_s\mu^2}{n^2}(\gamma^2\delta^2 \wedge \mu_s r_1)\lambda_{\max}^2\log^2 d + r\lambda_{\max}^2\delta^2 + \frac{d^*r(d_2\vee d_3r_1)}{n}\log d\cdot\sigma^2.
	\end{align*}
	In conclusion, we have 
	\begin{align*}
		&\quad \fro{\hat\bcalT - \bcalT^*}^2 \\
		&\leq 2\fro{\M_{l_{\max}} - \M^*}^2 + 2\fro{\calM_2^{-1}(\M^*)(\U_s-\U^*)}^2\\
		&\lesssim \frac{d_2d_3 rr^*\mu_s\mu^2}{n^2}(\gamma^2\delta^2 \wedge \mu_s r_1)\lambda_{\max}^2\log^2 d + r\lambda_{\max}^2\delta^2 + \frac{d^*r(d_2\vee d_3r_1)}{n}\log d\cdot\sigma^2.
	\end{align*}

	\section{Technical Lemmas}\label{sec:lemmas}
	\begin{lemma}[Lemma 4.1, \cite{wei2016guarantees}]\label{lemma:ptperp}
		Let $\M$ be a rank $r$ matrix, and $\TT$ be the tangent space of the rank $r$ matrix manifold at $\M$. Let $\X$ be another rank $r$ matrix. Then 
		\begin{align*}
			\fro{\calP_{\TT}^{\perp}\X}\leq \frac{1}{\lambda_{\min}(\X)}\op{\M-\X}\fro{\M-\X}. 
		\end{align*}
	\end{lemma}

	\begin{lemma}\label{lemma:YYLR}
		With probability exceeding $1-2d^{-10}$, the following hold for arbitrary $\A,\B\in\RR^{d_3r_1\times r}$, and $\C,\D\in\RR^{d_2\times r_1r_3}$:
		\begin{align*}
			&|\inp{(\calY^*\calY - \frac{n}{d^*}\calI)(\L^*\A^\top)}{\L^*\B^\top}| \lesssim \bigg(\frac{\mu_s r_1}{d_1}\frac{\mu r}{d_2}\log d + \sqrt{\frac{n}{d^*}\frac{\mu_s r_1}{d_1}\frac{\mu r}{d_2}\log d}\bigg)\fro{\A}\fro{\B},\\
			&|\inp{(\calY^*\calY - \frac{n}{d^*}\calI)(\C\R^{*\top})}{\D\R^{*\top}}|
			\lesssim \bigg(\frac{\mu_s r_1}{d_1}\frac{\mu r}{d_3}\log d + \sqrt{\frac{n}{d^*}\frac{\mu_s r_1}{d_1}\frac{\mu r}{d_3}\log d}\bigg)\fro{\C}\fro{\D}. 
		\end{align*}
	\end{lemma}
	\begin{proof}
		We only prove the first inequality and the other one can be proved in a similar fashion. 
		Notice 
		\begin{align*}
			\inp{(\calY^*\calY - \frac{n}{d^*}\calI)(\L^*\A^\top)}{\L^*\B^\top} &= \sum_{i=1}^n\inp{\Y_i}{\L^*\A^\top}\inp{\Y_i}{\L^*\B^\top}
		\end{align*}
		And recall $\Y_i^\top\L^* = (\I\otimes\U_s)^\top(\e_{g_i}\otimes \e_{l_i})\e_{k_i}^\top\L^*$. 
		We define $\calZ_i(\A) = \inp{\A}{\Y_i^\top\L^*}\Y_i^\top\L^*$.
		Then $\inp{\calY^*\calY (\L^*\A^\top)}{\L^*\B^\top}  = \sum_{i=1}^n\inp{\calZ_i(\A)}{\B}$.
		And the uniform bound on $\op{\calZ_i}$ is as follows
		\begin{align*}
			\op{\calZ_i} = \max_{\fro{\A}=1}\fro{\calZ_i(\A)}\leq \frac{\mu_s r_1}{d_1}\frac{\mu r}{d_2}.
		\end{align*}
		Meanwhile, since $\calZ_i$ is self-adjoint, we consider 
		\begin{align*}
			\EE\calZ_i^2(\cdot) &= \frac{1}{d^*}\sum_{l,k,g} \inp{\cdot}{\Y_i^\top\L^*}\fro{\Y_i^\top\L^*}^2\Y_i^\top\L^*\\
			&\leq \frac{1}{d^*}\frac{\mu_s r_1}{d_1}\frac{\mu r_2}{d_2}\calI,
		\end{align*}
		where $\calI:\RR^{r_2\times d_3r_1}\rightarrow\RR^{r_2\times d_3r_1}$ is the identity map. And thus 
		\begin{align*}
			\op{\sum_{i=1}^n\calZ_i^2}\leq \frac{n}{d^*}\frac{\mu_s r_1}{d_1}\frac{\mu r}{d_2}. 
		\end{align*}
		Thus we conclude with probability exceeding $1-d^{-10}$, 
		\begin{align*}
			\op{\sum_{i=1}^n(\calZ_i - \EE\calZ_i)} \lesssim \frac{\mu_s r_1}{d_1}\frac{\mu r}{d_2}\log d + \sqrt{\frac{n}{d^*}\frac{\mu_s r_1}{d_1}\frac{\mu r}{d_2}\log d},
		\end{align*}
		which leads to
		\begin{align*}
			|\inp{(\calY^*\calY - \frac{n}{d^*}\calI)(\L^*\A^\top)}{\L^*\B^\top}| \lesssim \bigg(\frac{\mu_s r_1}{d_1}\frac{\mu r}{d_2}\log d + \sqrt{\frac{n}{d^*}\frac{\mu_s r_1}{d_1}\frac{\mu r}{d_2}\log d}\bigg)\fro{\A}\fro{\B}. 
		\end{align*}
	\end{proof}

	We define for column orthonormal matrices $\U,\V\in\RR^{d\times r}$, 
	\begin{align*}
		d(\U,\V) = \op{\U\U^{\top} - \V\V^{\top}},  
		\quad d_{\rm F}(\U,\V)= \fro{\U\U^{\top} - \V\V^{\top}}.
	\end{align*}
	\begin{lemma}\label{lemma:pty}
		Suppose $\{(k_i,l_i)\}_{i=1}^n$ is a set of indices sampled independently and uniformly from $[d_2]\times [d_3r]$. 
		With probability exceeding $1-d^{-10}$, 
		\begin{align*}
			\op{\frac{d^*}{n}\calP_{\TT^*}\calY^*\calY\calP_{\TT^*} - \calP_{\TT^*}} \leq \frac{\mu\mu_s}{n}(r_1d_2r_3 + r_1r_2d_3)\log d +  \sqrt{\frac{\mu\mu_s}{n}(r_1d_2r_3 + r_1r_2d_3)\log d}.
		\end{align*}
		In particular, if $n\gtrsim \mu\mu_s(r_1d_2r_3 + r_1r_2d_3)\log d$, then 
		\begin{align*}
			\op{\frac{d^*}{n}\calP_{\TT^*}\calY^*\calY\calP_{\TT^*} - \calP_{\TT^*}} \leq \sqrt{\frac{\mu\mu_s}{n}(r_1d_2r_3 + r_1r_2d_3)\log d} \leq 0.01.
		\end{align*}
	\end{lemma}
	\begin{proof}
		Denote the operators $\calY_i(\cdot) =\inp{\cdot}{\Y_i}$ with $\Y_i = \e_{k_i}(\e_{g_i}\otimes\e_{l_i})^\top(\I\otimes \U_s)$ and
        \[
        \calZ_i = \frac{d^*}{n}\calP_{\TT^*}\calY_i^*\calY_i\calP_{\TT^*} - \frac{1}{n}\calP_{\TT^*}.
        \]
        Notice $\EE\calZ_i = 0$ and $\frac{d^*}{n}\calP_{\TT^*}\calY^*\calY\calP_{\TT^*} - \calP_{\TT^*} = \sum_{i=1}^n\calZ_i$. 
		
		\noindent\textit{Upper bound for $\op{\calZ_i}$. } Notice 
		\begin{align*}
			(\calP_{\TT^*}\calY_i^*\calY_i\calP_{\TT^*})^2(\M) = \fro{\calP_{\TT^*}\big(\e_{k_i}(\e_{g_i}\otimes\e_{l_i})^\top(\I\otimes \U_s)\big)}^2\cdot\calP_{\TT^*}\calY_i^*\calY_i\calP_{\TT^*}(\M).
		\end{align*}
		Recall $\calP_{\TT^*}(\cdot) = \V^*\V^{*\top}\cdot + (\I-\V^*\V^{*\top})\cdot(\W^*\W^{*\top}\otimes \I)$. And therefore
		\begin{align*}
			\fro{\calP_{\TT^*}\big(\e_{k_i}(\e_{g_i}\otimes\e_{l_i})^\top(\I\otimes \U_s)\big)}^2 \leq \frac{\mu r_2}{d_2}\cdot \frac{\mu_s r_1}{d_1} + \frac{\mu r_3}{d_3}\cdot \frac{\mu_s r_1}{d_1}.
		\end{align*}
		As a result, 
		\begin{align*}
			\op{\calP_{\TT^*}\calY_i^*\calY_i\calP_{\TT^*}} \leq \fro{\calP_{\TT^*}\big(\e_{k_i}(\e_{g_i}\otimes\e_{l_i})^\top(\I\otimes \U_s)\big)}^2\leq \frac{\mu r_2}{d_2}\cdot \frac{\mu_s r_1}{d_1} + \frac{\mu r_3}{d_3}\cdot \frac{\mu_s r_1}{d_1}.
		\end{align*}
		And 
		\begin{align*}
			\op{\calZ_i} \leq \frac{2d^*}{n}(\frac{\mu r_2}{d_2}\cdot \frac{\mu_s r_1}{d_1} + \frac{\mu r_3}{d_3}\cdot \frac{\mu_s r_1}{d_1}) = \frac{2\mu\mu_s}{n}(r_1d_2r_3 + r_1r_2d_3). 
		\end{align*}
		
		\noindent\textit{Upper bound for variance term $\op{\EE\sum_{i=1}^n\calZ_i^2}$. } Notice 
		$$\calZ_i^2 = (\frac{d^*}{n})^2(\calP_{\TT^*}\calY_i^*\calY_i\calP_{\TT^*})^2 -\frac{2d^*}{n^2}\calP_{\TT^*}\calY_i^*\calY_i\calP_{\TT^*} + \frac{1}{n^2}\calP_{\TT^*}.$$ 
		And therefore
		\begin{align*}
			\EE \calZ_i^2 = (\frac{d^*}{n})^2\EE(\calP_{\TT^*}\calY_i^*\calY_i\calP_{\TT^*})^2 -\frac{1}{n^2}\calP_{\TT^*}. 
		\end{align*}
		Meanwhile, 
		\begin{align*}
			\op{\EE(\calP_{\TT^*}\calY_i^*\calY_i\calP_{\TT^*})^2} \leq \frac{1}{d^*}\bigg(\frac{\mu r_2}{d_2}\cdot \frac{\mu_s r_1}{d_1} + \frac{\mu r_3}{d_3}\cdot \frac{\mu_s r_1}{d_1}\bigg).
		\end{align*}
		So we conclude 
		\begin{align*}
			\op{\EE\sum_{i=1}^n\calZ_i^2} \leq \frac{\mu\mu_s}{n}(r_1d_2r_3 + r_1r_2d_3).
		\end{align*}
		So we conclude using operator Bernstein inequality, with probability exceeding $1-d^{-10}$, 
		\begin{align*}
			\op{\frac{d^*}{n}\calP_{\TT^*}\calY^*\calY\calP_{\TT^*} - \calP_{\TT^*}} \lesssim \frac{\mu\mu_s}{n}(r_1d_2r_3 + r_1r_2d_3)\log d +  \sqrt{\frac{\mu\mu_s}{n}(r_1d_2r_3 + r_1r_2d_3)\log d}.
		\end{align*}
	\end{proof}

	\begin{lemma}\label{lemma:empirical}
		Let $s>0$, and $t = s + \log 12 + 2\log \log(d)$. 
		Then with probability exceeding $1-e^{-s}$, the following holds for all $\gamma_1\in[\frac{1}{d_2d_3r_1},1], \gamma_2\in[1, \sqrt{d_2\wedge d_3r_1}]$, 
		\begin{align*}
			\beta_n(\gamma_1,\gamma_2) &\lesssim \frac{\sqrt{\mu_s}r_1}{\sqrt{d_1}}\gamma_1\gamma_2\bigg(\sqrt{\frac{n}{d^*}(d_3r_1\vee d_2) \log d} + \sqrt{\frac{\mu_s r_1}{d_1}}\log d\bigg) \\
			&\quad + \sqrt{nt\frac{\mu_s r_1^2}{d_1d^*}\gamma_1^2} + \frac{\mu_s r_1^2}{d_1}\gamma_1^2t.
		\end{align*}
	\end{lemma}
	\begin{proof}
		We set
        \[
        \begin{aligned}
        \delta_{1,j} &= 2^{-j} \quad \text{for } j = 0,\cdots, j_0:= \lceil\log_2(d_2d_3r_1)\rceil,\\
        \delta_{2,k} &= 2^k, \quad k = 0,\cdots, k_0:=\lceil\log_2(\sqrt{d_2\wedge d_3r_1})\rceil.
        \end{aligned}
        \] 
		For each $j,k$, we first derive an upper bound for $\beta_n(\gamma_1,\gamma_2)$ with $\gamma_1 = \delta_{1,j},\gamma_2 = \delta_{2,k}$. 
		
		We define the random matrix which is uniformly distributed in $\{\e_k(\e_g\otimes \e_l)^\top(\I\otimes \U_s): (l,k,g)\in[d_1]\times [d_2]\times [d_3]\}$ by $\Y$, and let $\Y_1,\cdots, \Y_n$ be $n$ i.i.d. copies of $\Y$. Observe that 
		\begin{align*}
			\sup_{\M\in\KK_{\gamma_1,\gamma_2}}\big|\inp{\M}{\Y}^2 - \frac{1}{d^*}\fro{\M}^2\big| &\leq \sup_{\M\in\KK_{\gamma_1,\gamma_2}}\max_{l,k,g}\big|\inp{\e_k(\e_g\otimes \e_l)^\top}{\M(\I\otimes \U_s)^\top}^2\big|\\
			&\leq r_1\gamma_1^2 \cdot \frac{\mu_s r_1}{d_1} = \frac{\mu_s r_1^2}{d_1}\gamma_1^2,
		\end{align*}
		where the last inequality holds since each column of $(\I\otimes \U_s)^\top$ has at most $r_1$ non-zeros and is bounded by $\sqrt{\frac{\mu_s r_1}{d_1}}$ in $\ell_2$ norm, and $\linf{\M}\leq \gamma_1$. On the other hand, 
		\begin{align*}
			\sup_{\M\in\KK_{\gamma_1,\gamma_2}}\var\inp{\Y}{\M}^2\leq \sup_{\M\in\KK_{\gamma_1,\gamma_2}}\EE\inp{\Y}{\M}^4\leq \frac{\mu_s r_1^2}{d_1d^*}\gamma_1^2.
		\end{align*}
		Following Bousquet's version of Talagrand concentration inequality \citep[Theorem 3.3.9]{gine2021mathematical}, with probability exceeding $1-e^{-t}$ for any $t>0$, we have 
		\begin{align}\label{betan}
			\beta_n(\gamma_1,\gamma_2) \leq 2\EE\beta_n(\gamma_1,\gamma_2) + \sqrt{2nt\frac{\mu_s r_1^2}{d_1d^*}\gamma_1^2} + \frac{4}{3}\frac{\mu_s r_1^2}{d_1}\gamma_1^2t.
		\end{align}
		And we have 
		\begin{align}\label{Ebetan}
			\EE\beta_n(\gamma_1,\gamma_2) &\leq 2\EE\sup_{\M\in\KK_{\gamma_1,\gamma_2}}\big|\sum_{i=1}^n\epsilon_i\inp{\Y_i}{\M}^2\big|\notag\\
			&\leq 4\frac{\sqrt{\mu_s}r_1}{\sqrt{d_1}}\gamma_1\EE\sup_{\M\in\KK_{\gamma_1,\gamma_2}}\big|\sum_{i=1}^n\epsilon_i\inp{\Y_i}{\M}\big|\notag\\
			&\leq 4\frac{\sqrt{\mu_s}r_1}{\sqrt{d_1}}\gamma_1 \cdot \EE\op{\sum_{i=1}^n\epsilon_i\Y_i}\cdot\nuc{\M}
		\end{align}
		Here the first inequality follows from symmetric inequality, and $\epsilon_1,\cdots,\epsilon_n$ are i.i.d. Rademacher random variables; and the second inequality follows from contraction inequality and $|\inp{\Y_i}{\M}|\leq \frac{\sqrt{\mu_s}r_1}{\sqrt{d_1}}\gamma_1$. Using noncommutative Bernstein inequality (e.g. \cite[Theorem 4]{recht2011simpler}), we conclude for arbitrary $s>0$, 
		\begin{align*}
			\PP\bigg(\op{\sum_{i=1}^n \epsilon_i\Y_i}\geq s\bigg) \leq (d_2+d_3r_1) \exp\bigg(\frac{-s^2/2}{\frac{n}{d^*}(d_2\vee d_3r_1) + \frac{s}{3}\sqrt{\frac{\mu_s r_1}{d_1}}}\bigg).
		\end{align*}
		This leads to the upper bound for the expectation of the operator norm:
		\begin{align*}
			\EE\op{\sum_{i=1}^n \epsilon_i\Y_i} \leq \sqrt{\frac{n}{d^*}(d_3r_1\vee d_2) \log d} + \sqrt{\frac{\mu_s r_1}{d_1}}\log d.
		\end{align*}
		Together with \eqref{betan}, \eqref{Ebetan}, we conclude with probability exceeding $1-e^{-t}$, 
		\begin{align*}
			\beta_n(\gamma_1,\gamma_2) \lesssim \frac{\sqrt{\mu_s}r_1}{\sqrt{d_1}}\gamma_1\gamma_2\bigg(\sqrt{\frac{n}{d^*}(d_3r_1\vee d_2) \log d} + \sqrt{\frac{\mu_s r_1}{d_1}}\log d\bigg) + \sqrt{nt\frac{\mu_s r_1^2}{d_1d^*}\gamma_1^2} + \frac{\mu_s r_1^2}{d_1}\gamma_1^2t.
		\end{align*}
		Taking union bound over all $j,k$, with probability at least $1-3\log_2^2(d)e^{-t}$, the above holds for all $\gamma_1\in\{\delta_{1,0},\cdots,\delta_{1,j_0}\}, \gamma_2\in\{\delta_{2,0},\cdots,\delta_{2,k_0}\}$. Now for arbitrary $\gamma_1\in[\frac{1}{d_2d_3r_1},1], \gamma_2\in[1, \sqrt{d_2\wedge d_3r_1}]$, there exist $j,k$ such that $\gamma_1\in[\delta_{1,j-1},\delta_{1,j}]$, $\gamma_2\in[\delta_{2,k-1},\delta_{2,k}]$, such that 
		\begin{align*}
			\beta_n(\gamma_1,\gamma_2) \leq \beta_n(\delta_{1,j},\delta_{2,k}) &\lesssim \frac{\sqrt{\mu_s}r_1}{\sqrt{d_1}}\gamma_1\gamma_2\bigg(\sqrt{\frac{n}{d^*}(d_3r_1\vee d_2) \log d} + \sqrt{\frac{\mu_s r_1}{d_1}}\log d\bigg) \\
			&\quad + \sqrt{nt\frac{\mu_s r_1^2}{d_1d^*}\gamma_1^2} + \frac{\mu_s r_1^2}{d_1}\gamma_1^2t,
		\end{align*}
		where in the last inequality we implicitly use $\delta_{1,j}\leq 2\gamma_1$, $\delta_{2,k}\leq 2\gamma_2$. 
	\end{proof}
	
	\begin{lemma}\label{lemma:spiki-incoh}
		If $\fro{\M-\M^*}\leq \frac{1}{2}\lambda_{\min}$, and $\zeta\geq \linf{\M^*}$, then the incoherence of $\svd_r\big(\trim_{\zeta}(\M)\big)$ is bounded by $\frac{4d_2d_3r_1}{r}\frac{\zeta^2}{\lambda_{\min}^2}$. 
	\end{lemma}
	\begin{proof}
		We denote $\tilde\M = \trim_{\zeta}(\M)$. And write the compact svd $\svd_r(\tilde\M) = \tilde\L \tilde\S \tilde\R^\top$. 
		Since $\zeta\geq \linf{\M^*}$, we have $\fro{\tilde\M - \M^*}\leq \fro{\M - \M^*}$. Moreover, since $\rank(\M^*) = r$, we have $\fro{\svd_r\big(\trim_{\zeta}(\M)\big)-\M^*}\leq \fro{\tilde\M - \M^*}$. This indicates 
		\begin{align*}
			\fro{\svd_r\big(\trim_{\zeta}(\M)\big)-\M^*}\leq \frac{1}{2}\lambda_{\min}. 
		\end{align*}
		And thus $\op{\tilde\S^{-1}}\leq 2\lambda_{\min}^{-1}$. 
		Then we have 
		\begin{align*}
			\ltwo{\e_i^\top\tilde\L} &= \ltwo{\e_i^\top\tilde\M\tilde\R \tilde\S^{-1}} \leq \ltwo{\e_i^\top\tilde\M\tilde\R}\cdot\op{\tilde\S^{-1}}\\
			&\leq \sqrt{d_3r_1}\linf{\tilde\M}\cdot 2\lambda_{\min}^{-1}\\
			& = 2\sqrt{d_3r_1}\zeta\lambda_{\min}^{-1}. 
		\end{align*}
		And we can similarly show the incoherence of $\tilde\R$.
	\end{proof}

	\begin{lemma}[Remark 6.2, \cite{keshavan2010matrix}]\label{lemma:truncation}
		Let $\U,\X\in\RR^{d\times r}$ be column orthonormal and $\incoh(\U)\leq \mu_0$ and $d_{\rm F}(\U,\X)\leq \frac{1}{16\pi}$. 
		Let $\hat\X = \trunc(\X,\mu_0)$. 
		Then $\hat\X$ satisfies $$\incoh(\hat\X)\leq 3\mu_0, \quad d_{\rm F}(\hat\X,\U)\leq 4\pi \cdot d_{\rm F}(\X,\U).$$ 
	\end{lemma}
	
	\hspace{1cm}

	We define a bi-linear operator $\calL:\RR^{d_2\times d_3r_1}\times\RR^{d_2\times d_3r_1} \rightarrow \RR^{d_2\times d_3r_1^2}$ as 
	\begin{align*}
		\calL(\A,\B)(i;j,k,l) = \A(i;j,k)\B(i;j,l), \forall i\in[d_2],j\in[d_3],k,l\in[r_1]. 
	\end{align*}
	Notice here we use the multi-index $(j,k,l)\in[d_3r_1^2]$ and $(j,k),(j,l)\in[d_3r_1]$ to denote the column indices of matrices. 
	\begin{lemma}\label{lemma:YYPt}
		With probability exceeding $1-d^{-10}$, the following holds for arbitrary $\A,\B\in\RR^{d_2\times d_3r_1}$
		\begin{align*}
			|\inp{(\calY^*\calY - \frac{n}{d^*}\calI)\A}{\B}| \lesssim \bigg(\sqrt{\frac{n\mu_s r_1}{d^*d_1}(d_2\vee d_3r_1)\log d} + \frac{\mu_s r_1}{d_1}\log d\bigg)\cdot\nuc{\calL(\A,\B)}. 
		\end{align*}
	\end{lemma}
	\begin{proof}
		From the definition of $\calY$, we see
		\begin{align*}
			\inp{\calY^*\calY(\A)}{\B} &= \sum_{i=1}^n\inp{\A}{\Y_i}\inp{\B}{\Y_i}\\
			&=\sum_{i=1}^n\inp{\A}{\e_{k_i}(\e_{g_i}\otimes\e_{l_i})^\top(\I\otimes \U_s)}\inp{\B}{\e_{k_i}(\e_{g_i}\otimes\e_{l_i})^\top(\I\otimes \U_s)}\\
			&= \sum_{i=1}^n\e_{k_i}^{\top}\A(\I\otimes \U_s)^{\top}(\e_{g_i}\otimes\e_{l_i}) \cdot\e_{k_i}^{\top}\B(\I\otimes \U_s)^{\top}(\e_{g_i}\otimes\e_{l_i})\\
			&=\bigg\langle\calL(\A,\B), \sum_{i=1}^n\e_{k_i}\big(\e_{g_i}\otimes (\U_s^{\top}\e_{l_i}\otimes \U_s^{\top}\e_{l_i})\big)^{\top}\bigg\rangle.
		\end{align*}
		And thus 
		\begin{align*}
			&\quad|\inp{(\calY^*\calY - \frac{n}{d^*}\calI)\A}{\B}| \\
			&= \bigg|\bigg\langle\calL(\A,\B), \sum_{i=1}^n\e_{k_i}\big(\e_{g_i}\otimes (\U_s^{\top}\e_{l_i}\otimes \U_s^{\top}\e_{l_i})\big)^{\top} - \EE\e_{k_i}\big(\e_{g_i}\otimes (\U_s^{\top}\e_{l_i}\otimes \U_s^{\top}\e_{l_i})\big)^{\top}\bigg\rangle\bigg|\\
			&\leq \nuc{\calL(\A,\B)}\cdot \bigg\|\sum_{i=1}^n\e_{k_i}\big(\e_{g_i}\otimes (\U_s^{\top}\e_{l_i}\otimes \U_s^{\top}\e_{l_i})\big)^{\top} - \EE\e_{k_i}\big(\e_{g_i}\otimes (\U_s^{\top}\e_{l_i}\otimes \U_s^{\top}\e_{l_i})\big)^{\top}\bigg\|. 
		\end{align*}
		We next bound $ \bigg\|\sum_{i=1}^n\e_{k_i}\big(\e_{g_i}\otimes (\U_s^{\top}\e_{l_i}\otimes \U_s^{\top}\e_{l_i})\big)^{\top} - \EE\e_{k_i}\big(\e_{g_i}\otimes (\U_s^{\top}\e_{l_i}\otimes \U_s^{\top}\e_{l_i})\big)^{\top}\bigg\|$. 
		For notation simplicity, we denote $\Z_i:=\e_{k_i}\big(\e_{g_i}\otimes (\U_s^{\top}\e_{l_i}\otimes \U_s^{\top}\e_{l_i})\big)^{\top}$. 
		The uniform bound is as follows:
		\begin{align*}
			\op{\Z_i-\EE\Z_i} \leq 2\frac{\mu_s r_1}{d_1}. 
		\end{align*}
		And 
		\begin{align*}
			\EE\Z_i\Z_i^{\top} &= \frac{1}{d^*}\sum_{l,k,g}\e_{k}\big(\e_{g}\otimes (\U_s^{\top}\e_{l}\otimes \U_s^{\top}\e_{l})\big)^{\top}\big(\e_{g}\otimes (\U_s^{\top}\e_{l}\otimes \U_s^{\top}\e_{l})\big)\e_k^{\top}\\
			&=\frac{d_3}{d^*}\sum_{l}(\e_{l}^{\top}\U_s\U_s^{\top}\e_{l})^2\I\\
			&\leq \frac{d_3}{d^*}\frac{\mu_s r_1}{d_1}r_1\I. 
		\end{align*}
		On the other hand,
		\begin{align*}
			\EE\Z_i^{\top}\Z_i &= \frac{1}{d^*}\sum_{l,k,g}\big(\e_{g}\otimes (\U_s^{\top}\e_{l}\otimes \U_s^{\top}\e_{l})\big)\e_k^{\top}\e_{k}\big(\e_{g}\otimes (\U_s^{\top}\e_{l}\otimes \U_s^{\top}\e_{l})\big)^{\top}\\
			&= \frac{d_2}{d^*}\sum_{l}\big(\I\otimes (\U_s^{\top}\e_{l}\e_l^{\top}\U_s\otimes \U_s^{\top}\e_{l}\e_l^{\top}\U_s)\big)\\
			&\leq \frac{d_2}{d^*}\frac{\mu_s r_1}{d_1}\sum_l\big(\I\otimes (\I\otimes \U_s^{\top}\e_{l}\e_l^{\top}\U_s)\big)\\
			&= \frac{d_2}{d^*}\frac{\mu_s r_1}{d_1}\I. 
		\end{align*}
		In conclusion, we have 
		\begin{align*}
			\max\bigg\{\op{\EE\sum_{i=1}^n \Z_i\Z_i^{\top}}, \op{\EE\sum_{i=1}^n \Z_i^{\top}\Z_i}\bigg\} \leq \frac{n\mu_s r_1}{d^*d_1}(d_2\vee d_3r_1). 
		\end{align*}
		Using matrix Bernstein concentration inequality, we see with probability exceeding $1-d^{-10}$, 
		\begin{align*}
			\op{\sum_{i=1}^n \Z_i - \EE\Z_i} \lesssim \sqrt{\frac{n\mu_s r_1}{d^*d_1}(d_2\vee d_3r_1)\log d} + \frac{\mu_s r_1}{d_1}\log d. 
		\end{align*}
	\end{proof}

	\begin{lemma}\label{lemma:LAB}
		Let $\A,\B\in\RR^{d_2\times d_3r_1}$ be two rank $r$ matrices that admit the low rank decompositions $\A = \L_1\R_1^{\top}, \B = \L_2\R_2^{\top}$ such that $\L_1,\L_2\in\RR^{d_2\times r}, \R_1,\R_2\in\RR^{d_3r_1\times r}$. Then we have 
		\begin{align*}
			\nuc{\calL(\A,\B)} \leq r_1^{3/2}r&\cdot \min\big\{\twoinf{\L_1}^2\fro{\L_2}^2, \fro{\L_1}^2\twoinf{\L_2}^2\big\}\\
			&\quad \cdot\min\big\{\twoinf{\R_1}^2\fro{\R_2}^2,\fro{\R_1}^2\twoinf{\R_2}^2\big\}
		\end{align*}
	\end{lemma}
	\begin{proof}
		In order to represent $\calL(\A,\B)$, we need to consider the following sub-matrices of $\A,\B$:
		\begin{align*}
			\A&= [\A_1,\cdots,\A_{r_1}]= [\L_1\R_1^\top\P_1,\cdots,\L_1\R_1^{\top}\P_{r_1}]=:[\L_1\R_{1,1}^{\top},\cdots,\L_1\R_{1,r_1}^{\top}],\\
			\B&=[\B_1,\cdots,\B_{r_1}]=  [\L_2\R_2^{\top}\P_1,\cdots,\L_2\R_2^{\top}\P_{r_1}]=:[\L_2\R_{2,1}^{\top},\cdots,\L_2\R_{2,r_1}^{\top}].
		\end{align*}
		Here $\P_p = [0,\cdots,0,\underbrace{\I_{d_3}}_{p\text{-th block}},0,\cdots,0]^{\top} \in\RR^{d_3r_1\times d_3}$ extracts the columns. 
		Recall
		\begin{align*}
			\calL(\A,\B) &= \sum_{p,q=1}^{r_1} \calR_{p,q}(\A_p\odot\B_q)\\
			&= \sum_{p,q} \calR_{p,q}\bigg(\sum_{s,t=1}^r \big(\L_1(:,s)\R_{1,p}(:,s)^{\top}\big)\odot \big(\L_2(:,t)\R_{2,q}(:,t)^{\top}\big)\bigg)\\
			&= \sum_{p,q} \sum_{s,t} \calR_{p,q}\bigg(\big(\L_1(:,s)\odot\L_2(:,t)\big) \big(\R_{1,p}(:,s)\odot\R_{2,q}(:,t)\big)^{\top}\bigg).
		\end{align*}
		Here $\calR_{p,q}:\RR^{d_2\times d_3} \rightarrow \RR^{d_2\times d_3 r_1^2}$ puts the matrix into the corresponding location of a wider matrix. 
		And therefore
		\begin{align*}
			\nuc{\calL(\A,\B)}^2 &\leq \bigg(\sum_{p,q}\sum_{s,t} \ltwo{\L_1(:,s)\odot\L_2(:,t)}\cdot\ltwo{\R_{1,p}(:,s)\odot\R_{2,q}(:,t)}\bigg)^2\\
			&\leq r_1^2r^2 \sum_{p,q}\sum_{s,t} \ltwo{\L_1(:,s)\odot\L_2(:,t)}^2\cdot\ltwo{\R_{1,p}(:,s)\odot\R_{2,q}(:,t)}^2\\
			&= r_1^2r^2 \sum_{p,q}\sum_{s,t} \bigg(\sum_{i=1}^{d_2}\L_1^2(i,s)\L_2^2(i,t)\bigg) \cdot\bigg(\sum_{j=1}^{d_3}\R_{1,p}^2(j,s)\R_{2,q}^2(j,t)\bigg)\\
			&\leq r_1^2r^2\sum_{i}\sum_{s,t}\L_1^2(i,s)\L_2^2(i,t) \cdot \sum_{j}\sum_{p,q}\sum_{s,t}\R_{1,p}^2(j,s)\R_{2,q}^2(j,t). 
		\end{align*}
		Now notice the first term is bounded by 
		\begin{align*}
			\sum_{i}\sum_{s,t}\L_1^2(i,s)\L_2^2(i,t) \leq \min\big\{\max_i\ltwo{\L_1(i,:)}^2\fro{\L_2}^2, \fro{\L_1}^2\max_i\ltwo{\L_2(i,:)}^2\big\}.
		\end{align*}
		And the second term is bounded by 
		\begin{align*}
			\sum_{j}\sum_{p,q}\sum_{s,t}\R_{1,p}^2(j,s)\R_{2,q}^2(j,t) \leq \max_j\sum_{p,s}\R_{1,p}^2(j,s)\cdot\fro{\R_2}^2\leq r_1\max_{j}\ltwo{\R_1(j,:)}^2\fro{\R_2}^2. 
		\end{align*}
		Similarly, we may obtain 
		\begin{align*}
			\sum_{j}\sum_{p,q}\sum_{s,t}\R_{1,p}^2(j,s)\R_{2,q}^2(j,t) \leq r_1\min\big\{\max_{j}\ltwo{\R_1(j,:)}^2\fro{\R_2}^2,\fro{\R_1}^2\max_{j}\ltwo{\R_2(j,:)}^2\big\}.
		\end{align*}
	\end{proof}

	Next we borrow the lemma for matrix perturbation in \cite[Lemma 14]{shen2023computationally}. 
	\begin{lemma}[Matrix Perturbation]\label{lemma:perturbation}
		Let $\M\in\RR^{d_1\times d_2}$ be a rank $r$ matrix with its smallest non-zero singular value $\lambda_{\min}$. Let $\boldsymbol{\Delta}\in\RR^{d_1\times d_2}$ be such that $\fro{\boldsymbol{\Delta}}\leq \frac{\lambda_{\min}}{8}$. Then we have 
		\begin{align*}
			\fro{\svd_r(\M+\boldsymbol{\Delta}) - \M}^2 \leq \fro{\boldsymbol{\Delta}}^2 + C\frac{\fro{\boldsymbol{\Delta}}^3}{\lambda_{\min}}
		\end{align*}
		for some absolute constant $C>0$. 
	\end{lemma}

	\begin{lemma}\label{lemma:ydelta}
		If $n\gtrsim \mu\mu_s r_1\max\{d_2r_3,d_3r_2\}\log d$, then with probability exceeding $1-d^{-10}$, 
		\begin{align*}
			\op{\calY^*(\bdelta)}\lesssim C\sqrt{\frac{\mu_s\mu^2 r^*}{d^*}}\frac{\gamma\delta \wedge \sqrt{\mu_s r_1}}{\sqrt{d_1}}\lambda_{\max} \log d + 1.01\frac{n}{d^*}\lambda_{\max}\delta.
		\end{align*}
	\end{lemma}
	\begin{proof}
		Recall 
		\begin{align*}
			\delta_i &= \inp{\bcalX_i\times_1(\U^{*}- \U_s)^\top}{\bcalC^*\times_2\V^*\times_3\W^*} \\
			&= \inp{\bcalX_i}{\bcalC^*\times_1(\U^{*}- \U_s)\times_2\V^*\times_3\W^*}\\
			&= \inp{\e_{k_i}(\e_{g_i}\otimes\e_{l_i})^\top(\I\otimes (\U^* - \U_s))}{\M^*}.
		\end{align*}
		Therefore
		\begin{align*}
			\EE \delta_i\Y_i = \frac{1}{d^*}\M^*\big(\I\otimes (\U^*-\U_s)^{\top}\U_s\big),
		\end{align*} 
		and 
		\begin{align*}
			|\delta_i| &\leq \linf{\bcalC^*\times_1(\U^{*}- \U_s)\times_2\V^*\times_3\W^*}\\
			&\leq \lambda_{\max}\twoinf{\U^{*}- \U_s}\cdot\twoinf{\V^*}\cdot\twoinf{\W^*}.
		\end{align*}	
		Since $\twoinf{\U^{*}- \U_s}\leq \frac{\gamma\delta}{\sqrt{d_1}}$ and it also has another upper bound $\ltwo{\U^{*}- \U_s}\leq 2\sqrt{\frac{\mu_s r_1}{d_1}}$, we have
		\begin{align*}
			\twoinf{\U^{*}- \U_s} \leq \min\{ \frac{\gamma\delta}{\sqrt{d_1}}, 2\sqrt{\frac{\mu_s r_1}{d_1}}\}
		\end{align*}
		So 
		\begin{align}\label{upperbound:deltai}
			|\delta_i| \leq \lambda_{\max}\sqrt{\frac{\mu^2 r_2r_3}{d^*}}\min\{\gamma\delta, 2\sqrt{\mu_s r_1}\}. 
		\end{align}
		Next we bound $\op{\sum_{i=1}^n(\delta_i\Y_i - \EE\delta_i\Y_i)}$. The upper bound of $\op{\delta_i\Y_i - \EE\delta_i\Y_i}$ is as follows:
		\begin{align*}
			\op{\delta_i\Y_i - \EE \delta_i\Y_i} \leq 2 |\delta_i|\cdot \op{\e_{k_i}(\e_{g_i}\otimes\e_{l_i})^{\top}(\I\otimes \U_s)}
			\leq2 \sqrt{\frac{\mu_s\mu^2 r^*}{d^*}}\frac{\gamma\delta \wedge 2\sqrt{\mu_s r_1}}{\sqrt{d_1}}\lambda_{\max}. 
		\end{align*}
		On the other hand, 
		\begin{align*}
			\EE\delta_i^2\Y_i\Y_i^\top &= \frac{1}{d^*}\sum_{l,k,g} [\bcalC^*\times_1(\U^*-\U_s)\times_2\V^*\times_3\W^*]_{lkg}^2\e_{k}(\e_{g}\otimes\e_{l})^{\top}(\I\otimes \U_s\U_s^\top)(\e_{g}\otimes\e_{l})\e_{k}^{\top}\\
			&\leq \frac{1}{d^*} \sum_{l,k,g}  [\bcalC^*\times_1(\U^*-\U_s)\times_2\V^*\times_3\W^*]_{lkg}^2\cdot\frac{\mu_s r_1}{d_1}\e_k\e_k^\top \\
			&= \frac{1}{d^*}\frac{\mu_s r_1}{d_1} \sum_k\ltwo{\bcalC^*\times_1(\U^*-\U_s)\times_2\e_k^\top\V^*\times_3\W^*}^2\cdot\e_k\e_k^\top\\
			&\leq \frac{1}{d^*}\frac{\mu_s r_1}{d_1} \delta_{op}^2\lambda_{\max}^2\frac{\mu r_2}{d_2}\I_{d_2},
		\end{align*}
		where $\delta_{op} = \op{\U_s\U_s^\top - \U^*\U^{*\top}}\leq \delta$. 
		Similarly, we can show 
		\begin{align*}
			\EE\delta_i^2\Y_i^\top\Y_i \leq  \frac{\mu\mu_s r_1r_3}{d^*d_1d_3} \lambda_{\max}^2\delta_{op}^2 \cdot \I_{d_3r_1}.
		\end{align*}
		If we denote $\B_i = \delta_i\Y_i - \EE\delta_i\Y_i$, then we have 
		\begin{align*}
			\max\bigg\{\op{\EE\sum_{i=1}^n\B_i\B_i^\top}, \op{\EE\sum_{i=1}^n\B_i^\top\B_i}\bigg\} \leq
			n\frac{\mu\mu_s r_1}{d^*d_1}\delta_{op}^2\lambda_{\max}^2\max\{\frac{r_2}{d_2}, \frac{r_3}{d_3}\}. 
		\end{align*}
		From matrix Bernstein concentration inequality, we obtain with probability exceeding $1-d^{-10}$, 
		\begin{align*}
			\op{\sum_{i=1}^n (\delta_i\Y_i - \EE \delta_i\Y_i)} 
			\lesssim \sqrt{\frac{\mu_s\mu^2 r^*}{d^*}}\frac{\gamma\delta \wedge \sqrt{\mu_s r_1}}{\sqrt{d_1}}\lambda_{\max} \log d + \sqrt{n\frac{\mu\mu_s r_1}{d^*d_1}\delta_{op}^2\lambda_{\max}^2\max\{\frac{r_2}{d_2}, \frac{r_3}{d_3}\}\log d}. 
		\end{align*}
		And this leads to 
		\begin{align*}
			\op{\calY^*(\bdelta)} &\leq  C\bigg(\sqrt{\frac{\mu_s\mu^2 r^*}{d^*}}\frac{\gamma\delta \wedge \sqrt{\mu_s r_1}}{\sqrt{d_1}}\lambda_{\max} \log d + \sqrt{n\frac{\mu\mu_s r_1}{d^*d_1}\delta_{op}^2\lambda_{\max}^2\max\{\frac{r_2}{d_2}, \frac{r_3}{d_3}\}\log d}\bigg)\notag\\ 
			&\quad+ \frac{n}{d^*}\lambda_{\max}\delta. 
		\end{align*}
		As long as $n\gtrsim \mu\mu_s r_1\max\{d_2r_3,d_3r_2\}\log d$, 
		\begin{align*}
			\sqrt{n\frac{\mu\mu_s r_1}{d^*d_1}\delta_{op}^2\lambda_{\max}^2\max\{\frac{r_2}{d_2}, \frac{r_3}{d_3}\}\log d}\lesssim \frac{n}{d^*}\lambda_{\max}\delta. 
		\end{align*}
		And thus 
		\begin{align*}
			\op{\calY^*(\bdelta)} \leq C\sqrt{\frac{\mu_s\mu^2 r^*}{d^*}}\frac{\gamma\delta \wedge \sqrt{\mu_s r_1}}{\sqrt{d_1}}\lambda_{\max} \log d + 1.01\frac{n}{d^*}\lambda_{\max}\delta. 
		\end{align*}
	\end{proof}

	\begin{theorem}[Proposition 1, \cite{koltchinskii2011nuclear}]\label{thm:concentration:psi2}
		Let $\Z,\Z_1,\cdots, \Z_n$ be i.i.d. random matrices with dimensions $p_1\times p_2$ that satisfy $\EE\Z = 0$. Suppose $\psitwo{\op{\Z}}<+\infty$. Then there exists a constant $C>0$ such that for all $t>0$, with probability at least $1-e^{-t}$, 
		\begin{align*}
			\op{\frac{1}{n}\sum_{i=1}^n\Z_i} \leq C\max\bigg\{\sigma_Z\sqrt{\frac{t+\log p}{n}},\psitwo{\op{\Z}}\cdot\big(\log\frac{\psitwo{\op{\Z}}}{\sigma_Z}\big)^{1/2}\frac{t+\log p}{n} \bigg\},
		\end{align*}
		where $p = p_1+p_2$,
		\begin{align*}
			\sigma_Z = \max\bigg\{\bigg\|\frac{1}{n}\sum_{i=1}^n\EE\Z_i\Z_i^\top\bigg\|^{1/2}, \bigg\|\frac{1}{n}\sum_{i=1}^n\EE\Z_i^\top\Z_i\bigg\|^{1/2}\bigg\}.
		\end{align*}
	\end{theorem}

\end{document}